\documentclass[10pt,journal,compsoc]{IEEEtran}
\ifCLASSOPTIONcompsoc
  \usepackage[nocompress]{cite}
\else
  \usepackage{cite}
\fi
\ifCLASSINFOpdf
  \usepackage[pdftex]{graphicx}
  \graphicspath{{../pdf/}{../jpeg/}}
  \DeclareGraphicsExtensions{.pdf,.jpeg,.png}
\else
  \usepackage[dvips]{graphicx}
  \graphicspath{{../eps/}}
\fi
\usepackage{epstopdf}

\usepackage{amsmath}
\usepackage{amsthm}
\usepackage[noend]{algorithmic}
\usepackage{algorithm}

\ifCLASSOPTIONcompsoc
  \usepackage[caption=false,font=footnotesize,labelfont=sf,textfont=sf]{subfig}
\else
  \usepackage[caption=false,font=footnotesize]{subfig}
\fi
\usepackage{url}

\usepackage{ragged2e}
\usepackage{color,xcolor}

\newtheorem{theorem}{Theorem}
\newtheorem{definition}{Definition}

\newtheorem{remark}{Remark}

\begin{document}
%
\title{GriT-DBSCAN: A Spatial Clustering Algorithm for Very Large Databases}

%
%

\author{Xiaogang~Huang,
        Tiefeng~Ma,
        Conan~Liu,
        and Shuangzhe~Liu
\IEEEcompsocitemizethanks{\IEEEcompsocthanksitem
Xiaogang~Huang and Tiefeng~Ma are with School of Statistics, Southwestern University of Finance and
Economics, Chengdu 611130, China.\protect\\
E-mail: 1200202J8001@smail.swufe.edu.cn and matiefeng@swufe.edu.cn
\IEEEcompsocthanksitem
Conan~Liu is with UNSW Business School, University of New South Wales, Sydney, NSW 2052, Australia.\protect\\
E-mail: conan.liu88@gmail.com
\IEEEcompsocthanksitem
Shuangzhe~Liu is with Faculty of Science and Technology, University of Canberra, Canberra, ACT 2601, Australia.\protect\\
E-mail: Shuangzhe.Liu@canberra.edu.au
}

\thanks{Manuscript received April 6, 2022.}}

%
%

\markboth{Journal of \LaTeX\ Class Files}%
{Shell \MakeLowercase{\textit{et al.}}: A Linear Time Algorithm for Spatial Clustering}
%



\IEEEtitleabstractindextext{%
\begin{abstract}
\justifying DBSCAN is a fundamental spatial clustering algorithm with numerous practical applications.
However, a bottleneck of the algorithm is in the worst case, the run time complexity is $O(n^2)$.
To address this limitation, we propose a new grid-based algorithm for exact DBSCAN in Euclidean space called GriT-DBSCAN, which is based on the following two techniques.
First, we introduce a grid tree to organize the non-empty grids for the purpose of efficient non-empty neighboring grids queries.
Second, by utilising the spatial relationships among points, we propose a technique that iteratively prunes unnecessary distance calculations when determining whether the minimum distance between two sets is less than or equal to a certain threshold.
We theoretically prove that the complexity of GriT-DBSCAN is linear to the data set size.
In addition, we obtain two variants of GriT-DBSCAN by incorporating heuristics, or by combining the second technique with an existing algorithm.
Experiments are conducted on both synthetic and real-world data sets to evaluate the efficiency of GriT-DBSCAN and its variants.
The results of our analyses show that our algorithms outperform existing algorithms.
\end{abstract}

}

\maketitle

\IEEEdisplaynontitleabstractindextext

%
\IEEEpeerreviewmaketitle

\IEEEraisesectionheading{\section{Introduction}\label{sec:introduction}}
\IEEEPARstart{S}{patial} clustering is a fundamental technique in data analysis, and has extensive applications in data mining \cite{larose2014discovering}, information retrieval \cite{bhatia1998conceptual}, image segmentation \cite{ROSENFELD19821}, pattern recognition \cite{prml, jain2000statistical}, etc.
Generally speaking, the objective of spatial clustering is to divide a given data set into several clusters, such that objects in the same cluster are homogeneous, and objects from different clusters are heterogeneous.
A large number of algorithms have been proposed for spatial clustering.

Among all the spatial clustering algorithms, DBSCAN \cite{ester1996density} is perhaps one of the most widely used algorithms since it can discover clusters of arbitrary shapes and noises.
However, the worst case running time complexity of DBSCAN is $O(n^2)$ \cite{gunawan2013faster}, regardless of the parameters $\epsilon$ (radius) and $MinPts$ (density threshold).
This is because DBSCAN uses pointwise $\epsilon$-neighborhood queries.
When all points are within an $\epsilon$-ball, the running time for all of these queries is $O(n^2)$.
Many improved algorithms have been proposed to reduce the complexity of DBSCAN.

In general, there are three main strategies to reduce the complexity of DBSCAN: grid-based, ball-based, and sampling-based strategies.
The grid-based algorithms, such as G13 introduced by Gunawan and de Berg \cite{gunawan2013faster}, use the grid structure to reduce the range query time by considering the neighboring grids.
They proved that the running time of G13 is $O(n\log n)$ for 2-dimensional data.
G13 was extended by Gan and Tao \cite{gan2015dbscan, gan2017hardness} to higher dimensions with sub-quadratic complexity,
and a linear time approximate algorithm was proposed.
On the other hand, the ball-based algorithms, such as NQ-DBSCAN \cite{chen2018fast}, use  $2\epsilon$-ball (the $d$-dimensional ball with radius $2\epsilon$) to perform local neighborhood searching so as to reduce the time of range query.
However, the ball-based algorithms require $O(n^2)$ time in the worst case.
Unlike the two strategies just mentioned, the sampling-based algorithms improve DBSCAN by reducing the number of range queries.
For example, IDBSCAN \cite{borah2004improved} expands the cluster by performing range queries on some representatives sampled inside a core point's neighborhood.
Although sampling-based algorithms are faster, their results may be inconsistent with those of DBSCAN.
To the best of our knowledge, there is no algorithm that can produce the same results as DBSCAN does, while exhibiting complexity that is linear to the number of data points.

In this paper, we propose a simple and efficient grid-based DBSCAN algorithm.
The main contributions of this paper are as follows:
\begin{itemize}
\item
We introduce a novel tree-like data structure, namely, grid tree, to organize the non-empty grids.
Using the grid tree, we propose an efficient non-empty neighboring grids query technique.
For each non-empty grid, the running time to find its non-empty neighboring grids is linear to the number of non-empty neighboring grids in the best case.
\item
We propose an efficient technique, namely, FastMerging, to determine whether the minimum distance between two sets is less than or equal to a certain threshold, which iteratively prunes unnecessary distance calculations by considering the spatial relationships among points.
We also show that the complexity of FastMerging is linear to the number of points.
\item
Based on the above two techniques, we propose a new grid-based DBSCAN algorithm called GriT-DBSCAN, whose complexity is linear to the number of points.
\item
We conduct extensive experiments to evaluate the performance of GriT-DBSCAN and its two variants as well.
The results show the superiority of our algorithm and its variants.
\end{itemize}

The rest of this paper is organized as follows.
Section \ref{related work} presents the related work.
Section \ref{preliminaries} revisits the existing algorithms that are related to our proposed algorithm.
Section \ref{proposed algorithm} presents our proposed algorithm, including a detailed description of our algorithm and analysis of its theoretical properties.
Section \ref{experiments} compares the performance of our algorithm and its variants with those existing algorithms.
Finally, the conclusions are given in Section \ref{conclusions}.

\section{Related work} \label{related work}
In order to reduce the complexity of DBSCAN, many algorithms have been proposed.
This section briefly reviews the existing algorithms, that are related to our algorithm we propose in
Section \ref{proposed algorithm}.

\subsection{Grid-based Algorithms}
Grid-based algorithms try to reduce the complexity of DBSCAN by partitioning the feature space into grids so as to reduce the range query time by only considering the neighboring grids.
GriDBSCAN \cite{mahran2008using} partitions the feature space into equally sized grids.
The points in each grid and the points in the $\epsilon$-enclosure around the grid are considered a group.
Then, for each point in a certain grid, its $\epsilon$-neighbors are inside the corresponding group, which reduces the time for range query.
However, GriDBSCAN requires one more parameter given by the user.
Recently, Gunawan and de Berg \cite{gunawan2013faster} proposed a grid-based algorithm called G13, which in the worst case has run time complexity $O(n\log n)$ for 2-dimensional data, where the side length of each grid is $\epsilon/\sqrt2$.
G13 improves DBSCAN in three aspects.
First, if a grid contains at least $MinPts$ points, all points in that grid are core points because the distance between any two points in the grid is not greater than $\epsilon$.
This saves a lot of unnecessary distance calculations when identifying core points.
Second, by using the grid structure, only the points in the neighboring grids are considered when finding the $\epsilon$ neighbors of a certain point.
Third, the Voronoi diagram is used to merge the core grids, which dominates the complexity of G13.
Nonetheless, G13 is only suitable for 2-dimensional data.
Gan and Tao \cite{gan2015dbscan, gan2017hardness} extended G13 to higher dimensions.
To be specific, Gan and Tao proposed an exact DBSCAN algorithm that runs in sub-quadratic time and an approximate DBSCAN algorithm that runs in $O(n)$ expected time.
When merging the core grids, the approximate DBSCAN algorithm proposed by Gan and Tao exploits a quadtree-like hierarchical grid structure for approximate range count, resulting in an expected complexity of $O(n)$ for the entire algorithm.
As pointed out in Boonchoo et al. \cite{boonchoo2019grid}, that different order of forming clusters in the exact DBSCAN algorithm proposed by Gan and Tao \cite{gan2015dbscan, gan2017hardness} leads to different running time, so they proposed two strategies, namely, uniform random order and low density first order, to reduce the running time of forming clusters.
However, the results of \cite{boonchoo2019grid} do not improve the theoretical running time of previous work.
To further improve the running time, Wang et al. \cite{wang2020theoretically} parallelized the algorithms proposed by Gan and Tao \cite{gan2015dbscan, gan2017hardness}, and recommended new parallel algorithms for exact DBSCAN and approximate DBSCAN that match the complexity of their sequential counterparts.

To summarise, there is no known grid-based DBSCAN that can obtain the same results as DBSCAN while exhibiting complexity linear to the data set size.

\subsection{Ball-based Algorithms}
Ball-based algorithms divide the data set into several subsets using d-dimensional ball.
Both AnyDBC \cite{mai2016anydbc} and IncAnyDBC \cite{mai2020incremental} use $\epsilon$-ball to divide the data set into primitive clusters.
Then, the cluster structure of the data set is iteratively and actively learned.
In each iteration, a small set of most promising points are selected for refining clusters.
G-DBSCAN \cite{kumar2016fast} employs the group method based on $\epsilon$-ball to obtain a set of groups and runs DBSCAN using groups to accelerate range queries.
NQ-DBSCAN \cite{chen2018fast} uses a local neighborhood searching technique based on $2\epsilon$-ball to reduce the time of range query.
By using $k$NN-ball ($k$-nearest neighbors ball), KNN-BLOCK DBSCAN \cite{chen2019knn} can quickly identify the core points and partition the data set into core-blocks, noncore-blocks, and noise-blocks.
Then the core blocks are merged into clusters.
It is worth pointing out that KNN-BLOCK DBSCAN is an approximate algorithm because it utilizes a fast approximate $k$NN algorithm, namely, FLANN \cite{muja2014scalable}, to find $k$-nearest neighbors.
Moreover, the complexity of KNN-BLOCK DBSCAN is $O(n^2)$ if there is no core-block, and $k$-nearest neighbors queries runs in $O(n)$.
Recently, BLOCK-DBSCAN \cite{CHEN2021107624} employs $\epsilon /2$-ball to quickly identify inner core blocks, within which all points are core points, and applies a fast approximate algorithm to classify whether two inner core blocks are density reachable from each other.
However, when there is no inner core block and the complexity of finding the $\epsilon$-neighborhood of each point is $O(n)$, the complexity of BLOCK-DBSCAN is $O(n^2)$.

Although the ball-based algorithms improve the performance of DBSCAN to some extent, their worst-case time complexity is still $O(n^2)$ \cite{kumar2016fast, chen2018fast, CHEN2021107624, chen2019knn, mai2020incremental}.

\subsection{Sampling-based Algorithms}
DBSCAN expands clusters by performing range query for each point in the data set, which is particularly time consuming.
The sampling-based algorithms are designed to perform range queries on a subset of the data set or reduce the range query time by finding $\epsilon$-neighbors on a subset , thereby reducing the complexity of DBSCAN.
SDBSCAN \cite{zhou2000combining} selects a small number of representative points from the data set to perform DBSCAN.
$l$-DBSCAN \cite{viswanath2006dbscan} obtains prototypes from the data set by leaders algorithm \cite{hartigan1975clustering} and runs DBSCAN on the prototypes to form clusters. Rough-DBSCAN \cite{viswanath2009rough} also uses the leaders algorithm to select representative points, but uses rough set theory \cite{pawlak1991rough} to analyze the quality of clustering results.
Unlike the sampling-based algorithms mentioned above, IDBSCAN \cite{borah2004improved} reduces the number of range queries by selecting representatives inside the $\epsilon$-neighborhood of each core point when expanding clusters.
FDBSCAN \cite{liu2006fast} omits the unnecessary range queries by selecting representatives outside the $\epsilon$-neighborhood of a core point.
NG-DBSCAN \cite{lulli2016ng} is an approximate density-based algorithm which first creates an approximate $k$NN graph by computing the distance of points at a 2-hop distance, and then utilizes the approximate $k$NN graph to form clusters.
Recently, Jang and Jiang \cite{jang2019dbscan} suggested DBSCAN++, which reduces the complexity of DBSCAN by performing range query on a subset of the data set obtained by uniform sampling or k-center algorithm \cite{gonzalez1985clustering}.
For each point, SNG-DBSCAN \cite{jang2020dbscan} finds its $\epsilon$-neighborhood on a subset of the data set to reduce the range query time.

Nevertheless, the sampling-based improved algorithms still have the drawback that the clustering results may be inconsistent with the results of DBSCAN.

\section{Preliminaries} \label{preliminaries}
This section covers the definitions and basic clustering processes of the existing algorithms that are related to our proposed algorithm, which facilitates the comprehension of our proposed algorithm.
Section \ref{DBSCAN-original} summarizes the key notations used throughout this paper in Table \ref{NOTATIONS}, and reviews the original DBSCAN algorithm in \cite{ester1996density}.
Section \ref{DBSCAN-2D} reviews the G13 proposed by Gunawan and de Berg \cite{gunawan2013faster} that improves the performance of DBSCAN in the 2-dimensional space.
Section \ref{appr-dbscan} reviews the state of the art approximate algorithm---$\rho$-approximate DBSCAN in \cite{gan2015dbscan},\cite{gan2017hardness}---that solves the problem with slight inaccuracy in $O(n)$ expected time.

\begin{table}[!t]
\renewcommand{\arraystretch}{1.3}
\caption{Notations used in this paper}
\label{NOTATIONS}
\centering
\begin{tabular}{l l}
\hline
\bfseries Notation & \bfseries Description\\
\hline
$P$ & Data set\\
$p,q$ & The points\\
$n$ & Number of points\\
$d$ & Number of dimensions\\
$dist(\cdot,\cdot)$ & Distance function\\
$MinDist(\cdot,\cdot)$ & Minimum distance between two sets\\
$\epsilon,MinPts$ & Parameters of DBSCAN\\
$N_{\epsilon}(p)$ & $p$'s $\epsilon$-neighborhood\\
$G$ & An undirected graph\\
$Gs$ & The set of non-empty grids\\
$g$ & The grid\\
$\mathcal{N}_{\epsilon}(g)$ & The set of $g$'s neighboring grids \\
$Nei(g)$ & The set of $g$'s non-empty neighboring grids \\
$m$ & Number of core points in a grid\\
$\eta$ & Maximum interval number in the feature space\\
$T$ & A grid tree\\
$root$ & The root of $T$\\
$t$ & A node in $T$\\
$h$ & The hash table of $T$\\
$\gamma,\omega,\theta,\lambda$ & The angles\\
$\kappa$ & Number of iterations\\
\hline
\end{tabular}
\end{table}

\subsection{DBSCAN} \label{DBSCAN-original}
Consider a set of $n$ points $P$ in a $d$-dimensional space with a distance function $dist$: $R^{d} \times R^{d} \to R$ giving the Euclidean distance $dist(p, q)$ between $p, q \in P$.
$N_{\epsilon}(p)$ denotes the $\epsilon$-neighborhood of $p \in P$ with radius $\epsilon$, i.e. $N_{\epsilon}(p) = \{q|dist(p, q) \le \epsilon, q \in P\}$.

Some important concepts of DBSCAN are defined as follows.
\begin{definition}
	\label{corepoint}
    A point $p \in P$ is a core point if it's $\epsilon$-neighborhood satisfies $|N_{\epsilon}(p)| \ge MinPts$.
\end{definition}

\begin{definition}
	\label{ddr}
    A point $p \in P$ is directly density-reachable from a point $q \in P$ wrt. $\epsilon, MinPts$ if $q$ is a core point and $p \in N_{\epsilon}(q)$.
\end{definition}

Note that point $p$ does not need to be a core point in Definition \ref{ddr}. Consequently, direct density-reachability is not symmetric if $p$ is a non-core point.

\begin{definition}
	\label{dr}
    A point $p \in P$ is density-reachable from a point $q \in P$ wrt. $\epsilon, MinPts$ if there is a sequence of points $p_1, p_2, ..., p_z \in P$ such that $p_1 = q, p_z = p$, and $p_{i+1}$ is directly density-reachable from $p_{i}$, where $i=1, 2, ..., z - 1$.
\end{definition}

\begin{definition}
	\label{dc}
    A point $p \in P$ is density-connected to a point $q \in P$ wrt. $\epsilon, MinPts$ if there is a point $b \in P$ such that $p$ and $q$ are density-reachable from $b$ wrt. $\epsilon, MinPts$.
\end{definition}

\begin{definition}
	\label{cluster}
    A non-empty set $C$ is a cluster wrt. $\epsilon, MinPts$ if $C$ satisfies the following conditions:
    \begin{enumerate}[\IEEEsetlabelwidth{12)}]
    \item If $p \in C$ and $p$ is a core point, then all points density-reachable from $p$ wrt. $\epsilon, MinPts$ also belong to $C$. (Maximality)
    \item $\forall p, q \in C$, $p$ is density-connected to $q$ wrt. $\epsilon, MinPts$. (Connectivity)
    \end{enumerate}
\end{definition}

Given two parameters $\epsilon$ and $MinPts$, DBSCAN algorithm aims to find all clusters in $P$.
DBSCAN algorithm randomly selects an unprocessed core point from $P$, and then finds all points that are density-reachable from $p$ wrt. $\epsilon, MinPts$ to form a cluster.
The algorithm repeats the above step until all clusters are found.

A cluster found by DBSCAN contains not only core points but also non-core points.
A non-core point is called a border point if it belongs to at least one cluster.
And points do not belong to any cluster are called noise points.

\subsection{G13} \label{DBSCAN-2D}
G13 solves the exact DBSCAN with $O(n\log n)$ time complexity using the grid structure.
The algorithm consists of four major steps.
First, partition the feature space into equal-sized grids.
Second, identify core points in the data set.
Third, merge core points to form clusters.
Lastly, assign non-core points to clusters according to density reachability.

In the first step, the feature space is divided into multiple grids of the same size with side length $\epsilon/ \sqrt{2}$, and then each point will be assigned to the grid it lies in.
A grid is called an empty grid if there are no points inside it. Otherwise, it is a non-empty grid.

The second step is to identify all core points in the data set.
A grid containing at least one core point is called a core grid.
Recall that the side length of each grid is $\epsilon/ \sqrt{2}$, which ensures the distance between any two points in the same grid will not be greater than $\epsilon$.
Therefore, if there are more than $MinPts$ points in the same grid, then all those points are core points.

For a grid $g_i$ with points less than $MinPts$, the algorithm checks each point $p$ in $g_i$ and determines whether it is a core point.
Let $\mathcal{N}_{\epsilon}(g_i)$ denote the neighboring grids of $g_i$, where $\mathcal{N}_{\epsilon}(g_i)=\{g^\prime|ditance(g^\prime, g_i) \le \epsilon\}$ and $ditance(g^\prime, g_i)$ is the minimum distance between $g^\prime$ and $g_i$.
To determine whether $p$ is a core point or not, the algorithm detects its neighbors in the non-empty grids of $\mathcal{N}_{\epsilon}(g_i)$.

The third step, namely the merging step, is to find clusters formed by core points.
To seek out clusters, we need to iterate through all core grids.

Initially, each core grid can be treated as an individual cluster based on the fact that all points in the same core grid are density-reachable from each other and the definition of a cluster.

Then those core grids that are density-reachable from each other will be merged into one cluster.
Let $G = (V, E)$ be a graph, where each vertex $v \in V$ corresponds to a core grid, and each edge $(g_i, g_j) \in E$ represents that $g_i$ and $g_j$ can be merged.
Given two core grids $g_i$ and $g_j$, the algorithm determines whether they can be merged by the following definition.

\begin{definition}
	\label{merge}
    For two core grids $g_i$ and $g_j$, they can be merged if and only if there are core points $p \in g_i$ and $q \in g_j$ such that $dist(p, q) \leq \epsilon$.
\end{definition}

Given a core grid $g_i$, for each core grid $g_j\in \mathcal{N}_{\epsilon}(g_i)$, we explain in detail how to verify whether $G$ has an edge between $g_i$ and $g_j$. For each core point $p \in g_i$, G13 uses the Voronoi diagram to find the core point $q \in g_j$ closest to $p$.
If the distance between $p$ and $q$ is not greater than $\epsilon$, then edge $(g_i, g_j)$ can be added to $G$.
If all core points $p \in g_i$ have been traversed and edge $(g_i, g_j)$ has not been added to $G$, then $g_i$ and $g_j$ can not be merged according to Definition \ref{merge}.
After the graph $G$ has been created, each connected component of $G$ represents a cluster.

The last step is to assign non-core points to clusters according to density reachability.
For each non-core point $p$, if there is at least one core point in $p$'s $\epsilon$-neighborhood, $p$ is called a border point. Otherwise, it is a noise point.

As shown in \cite{gunawan2013faster}, the time complexity of the merging step is $O(n\log n)$ if the Voronoi diagram is used, and the rest of the algorithm's time complexity is $O(MinPts \cdot n)$. Hence, the overall complexity of G13 is $O(n\log n)$ which is dominated by the merging step.

\subsection{$\rho$-approximate DBSCAN} \label{appr-dbscan}
$\rho$-approximate DBSCAN is an approximate algorithm that extends G13 to higher dimensions, and its complexity is linear to the data set size.

To ensure the distance between any two points in the same grid is not greater than $\epsilon$, the side length of the grid is set to be $\epsilon/ \sqrt{d}$.
In addition to the different setting of the side length of grid, the rule of merging core grids is different as well.
$\rho$-approximate DBSCAN uses a quadtree-like hierarchical grid structure for approximate range count query when merging core grids.
Formally, we fix a core grid $g_i$ and check each non-empty neighboring grid $g_j$ of $g_i$.
An edge $(g_i, g_j)$ is added to $G$ if there is a core point $p \in g_i$ and the approximate range count query on the core points of $g_j$ return yes.
Gan and Tao \cite{gan2015dbscan, gan2017hardness} showed that the overall complexity of $\rho$-approximate DBSCAN is $O(n)$ in expectation, regardless of the constant dimension $d$, the constant approximation ratio $\rho$, and the parameter $\epsilon$.

\section{The proposed algorithm} \label{proposed algorithm}
In this section, we propose our algorithm which aims at solving the exact DBSCAN with time complexity linear to the data set size.
Firstly, we present a new approach to construct grids.
Secondly, we introduce grid tree which allows us to efficiently find the non-empty neighboring grids for each grid.
Thirdly, we introduce a critical new technique to swiftly determine whether any two core grids should be merged.
Finally, we combine the above techniques to obtain our proposed algorithm called GriT-DBSCAN, which stands for \textbf{Gri}d \textbf{T}ree \textbf{DBSCAN}.

\subsection{Constructing the Grids}
\begin{figure}[!t]
\centering
\includegraphics[width=2.5in]{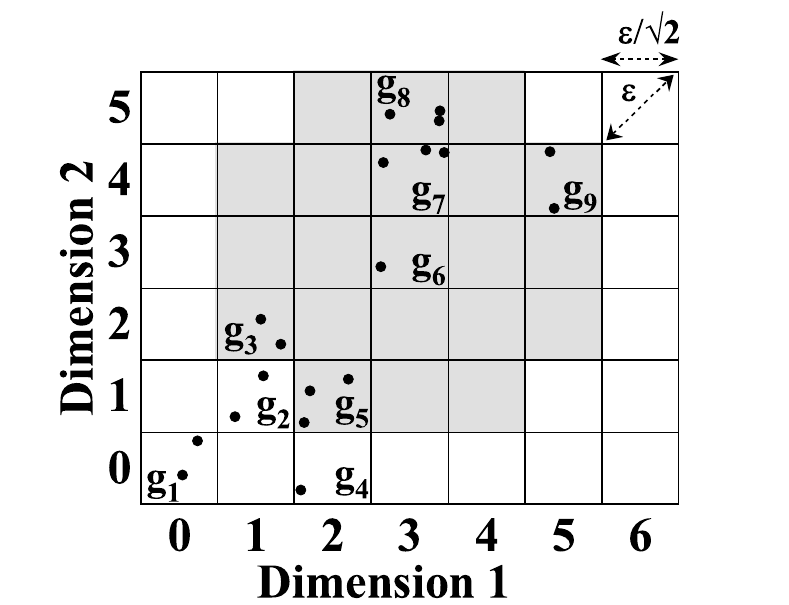}
\caption{Grids in 2-dimensional space ($\mathcal{N}_{\epsilon}(g_6)$ is shown in gray).}
\label{constructGrids}
\end{figure}

When constructing the grids, each dimension of the feature space is divided into equal size intervals of length $\epsilon/ \sqrt{d}$.
Then, the feature space is divided into multiple equally sized grids and each grid $g_i$ is uniquely determined by a $d$-dimensional vector $(g_{i1}, g_{i2}, ..., g_{ij}, ..., g_{id})\in N^d$  denoted as identifier.
A point $p$ lies in $g_i$ if and only if for all $j = 1,2,...,d$, the following formula holds:
\begin{equation}
\label{grid}
    g_{ij} = \lfloor\frac{p_{j}-mn_j}{\epsilon/ \sqrt{d}}\rfloor
\end{equation}
where $mn_j = \min_{p\in P}p_{j}$ and $\lfloor\cdot\rfloor$ is the floor function.

We can now cover how to construct grids.

First, for each point $p\in P$, the identifier of the grid it lies in is calculated by (\ref{grid}).

Second, the $n$ points are sorted using radix sort\cite{kunth1975art} according to their identifiers.
Initially, the $n$ points are sorted using counting sort according to the value of their identifiers on dimension $d$.
Then, the $n$ points are sorted according to the value of their identifiers on dimension $d-1$.
Repeat this process for each dimension.
After sorting, the points in the same grid are placed adjacently.

Finally, we traverse the sorted points and find all non-empty grids denoted as $Gs$.
Moreover, for any two grids $g_i,g_j \in Gs$ with $i<j$, there exists an integer $z\in [1, d)$ such that $g_{iz} < g_{jz}$ and $g_{iw} = g_{jw}$, for each $w\leq z-1$.

The detailed partitioning technique is shown in Algorithm \ref{partition}.

\begin{algorithm}[ht]
	\caption{Partitioning}
	\label{partition}
	\begin{algorithmic}[1]
        \REQUIRE point set $P$; $\epsilon$.
        \ENSURE $Gs$: non-empty grids.
        \STATE Calculate the identifiers by (\ref{grid}).
        \STATE Sort $P$ based on their identifiers using radix sort.
        \STATE Initialize $Gs$ as an empty set.
        \STATE $p = $ the 1st point in $P$
        \STATE $gs=\{p\}$
		\FOR{$i = 2:n$}
        \STATE $q = $ the $i$th point in $P$
        \IF {$q$ and $p$ not in the same grid}
        \STATE $Gs = Gs \cup \{gs\}$
        \STATE $gs=\emptyset$
        \ENDIF
        \STATE $p=q$
        \STATE $gs=gs \cup \{p\}$
        \ENDFOR
        \STATE $Gs = Gs \cup \{gs\}$
        \RETURN $Gs$
	\end{algorithmic}
\end{algorithm}

\begin{figure*}[!t]
\centering
\includegraphics[width=4.5in]{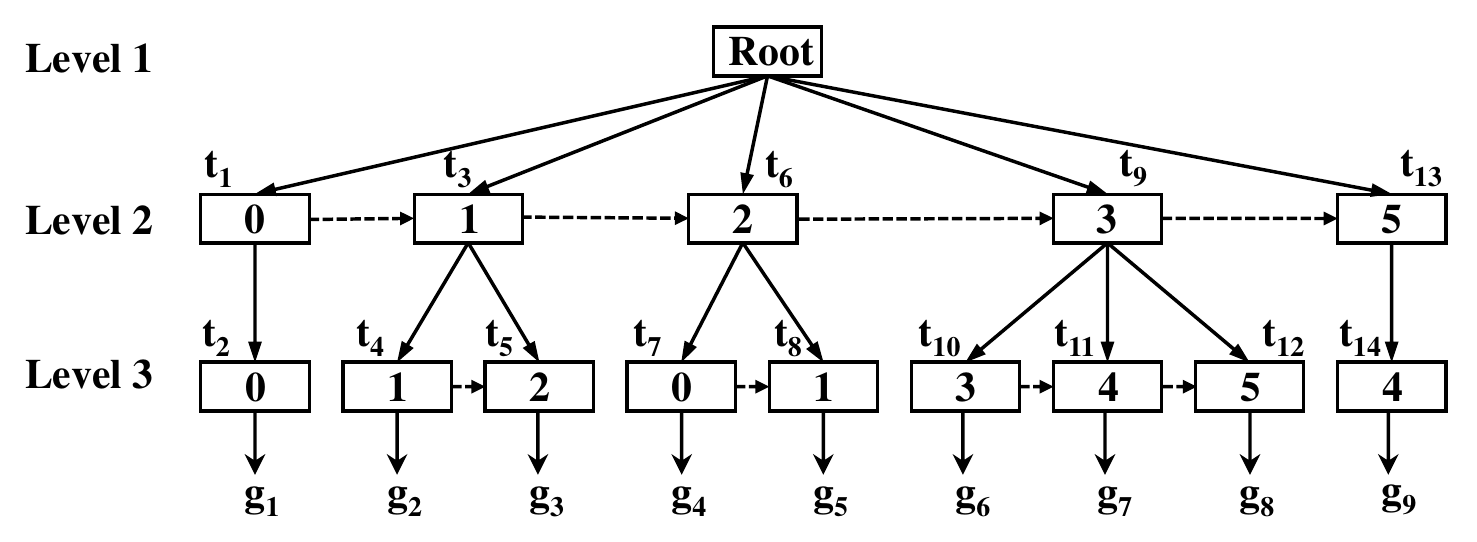}
\caption{Grid tree for the non-empty grids in Fig. \ref{constructGrids}.}
\label{gridTree}
\end{figure*}

It's easy to verify that the second step runs in $O(n+\eta)$, where $\eta=\max \{g_{ij}| g_i\in Gs,1\leq j\leq d\}$ is a constant associated with $\epsilon$, $d$, the maximum and minimum coordinates in each dimension.
The rest of the algorithm's time complexity is $O(n)$.
Therefore, the complexity of Algorithm \ref{partition} is $O(n+\eta)$.

To illustrate, Fig. \ref{constructGrids} shows the grid structure on a set with 19 points, where $\eta = 5$.
There are 9 non-empty grids in the feature space, and the set of non-empty grids is $Gs = \{g_1,g_2,...,g_9\}$.
In addition, the neighboring grids of $g_6$ are shown in gray.

\subsection{Indexing the Grids} \label{datastructure}
G13 and $\rho$-approximate DBSCAN significantly improve the performance of DBSCAN using the grid structure.
However, for a given grid $g_i$, the number of neighboring grids of $g_i$ increases exponentially with the dimension\cite{boonchoo2019grid}.
One important innovation of this paper is to introduce the grid tree as a structure to organize the non-empty grids $Gs$.
Given a non-empty grid $g_i$, the grid tree enables us to find all non-empty grids in $\mathcal{N}_{\epsilon}(g_i)$ more efficiently.
We first describe the definition of the grid tree.
Then, we introduce an algorithm for finding non-empty neighboring grids using the grid tree.
Finally, we present theoretical evidence showing the grid tree is an effective structure for non-empty neighboring grids queries.

\subsubsection{Grid Tree}
Grid tree is a tree-like structure built to organize non-empty grids.
Denoted as $T$, a grid tree has $d+1$ levels.
At the 1st level, there is a root node ($root$) which contains child pointers to child nodes.
And the $j$th child pointer of $root$ will be denoted as $CHILD_j(root)$.
All nodes in the $i$th level, where $1<i\leq d$, are called internal nodes.
In addition to the child pointers, each internal node $t$ contains one key denoted as $KEY(t)$, and a next pointer $NEXT(t)$ which is either null or points to another node in the same level. If the next pointer is not null, then we have $KEY(t) < KEY(NEXT(t))$.
Each leaf node $t$ is in the ${(d+1)}{th}$ level, containing a key, a next pointer, and a pointer to a non-empty grid denoted as $GRID(t)$.
For each node $t$ in the grid tree, if it has more than $MinPts$ child nodes, then there are many items associated with $t$ in the hash table $h$ of $T$ that can be used to tell which child node to visit at the next level when searching the non-empty neighboring grids for a particular grid.
Let $\lceil \cdot \rceil$ denote the ceiling function.
Each item of $h$ is a $\{t, key\}-pt$ pair, where $pt$ is the node with the smallest key among the child nodes of $t$ whose keys are between $key-\lceil \sqrt{d}\rceil$ and $key+\lceil \sqrt{d}\rceil$.
Furthermore, for each path $t_{(1)},t_{(2)},...,t_{(d+1)}$ from the root node to leaf node, where $t_{(1)}$ is the root node and $t_{(i+1)}$ is a child of $t_{(i)}$, $i=1,2,...,d$, the identifier of the non-empty grid that $t_{(d+1)}$ points to is equal to $(KEY(t_{(2)}),KEY(t_{(3)}),...,KEY(t_{(d+1)}))$.

\begin{algorithm}[ht]
	\caption{BuildingGridTree}
	\label{bgt}
	\begin{algorithmic}[1]
        \REQUIRE $Gs$: non-empty grids.
        \ENSURE $T$: a grid tree.
        \STATE Initialize $T$ with a root node and an empty hash table.
        \FOR{$i = 1:|Gs|$}
        \STATE $t=T.root$
        \FOR {$j=1:d$}
        \STATE $tc = $ $t$'s last child
        \IF {$tc = null$ \textbf{or} $KEY(tc)\neq g_{ij}$}
        \STATE /*Generate a path from $t$ to leaf node that identifies $g_i$.*/
        \FOR {$l=j:d$}
        \STATE Create a new node $nd$ with key $g_{il}$.
        \STATE $t$'s last child $=nd$
        \IF {$tc \neq null$ \textbf{and} $l=j$}
        \STATE $NEXT(tc) = nd$
        \ENDIF
        \STATE $t = nd$
        \ENDFOR
        \STATE break
        \ENDIF
        \STATE $t=tc$
        \ENDFOR
        \ENDFOR
        \STATE $h = T.h$
        \FORALL{node $t\in T$}
        \IF {$t$ has more than $MinPts$ child nodes}
        \STATE /*Insert items associated with $t$ into $h$.*/
        \STATE $nd=CHILD_1(t)$; $pos = 0$
        \WHILE {$nd\neq null$}
        \STATE $pos=\max\{pos,KEY(nd) - \lceil \sqrt{d}\rceil\}$
        \WHILE {$pos \leq KEY(nd) + \lceil \sqrt{d}\rceil$ \textbf{and} $pos\leq \eta$}
        \STATE Insert $\{t, pos\}-nd$ into $h$.
        \STATE $pos = pos+1$
        \ENDWHILE
        \STATE $nd = NEXT(nd)$
        \ENDWHILE
        \ENDIF
        \ENDFOR
        \RETURN $T$
	\end{algorithmic}
\end{algorithm}

Algorithm \ref{bgt} describes a method to construct a grid tree for non-empty grids.
First, we initially create a tree with one node.
Then, for each non-empty grid in $Gs$, nodes are inserted to form a path from the root node to leaf node which identifies the grid.
Finally, for each node in the grid tree, if it has more than $MinPts$ child nodes, then the algorithm will generate items associated with it into $h$ (Lines 17-27).

To illustrate, Fig. \ref{gridTree} shows the grid tree for the non-empty grids in Fig. \ref{constructGrids}.
Suppose that $MinPts = 3$. Then there are six items associated with the root node in $h$: $\{root, 0\}-t_1$, $\{root, 1\}-t_1$, $\{root, 2\}-t_1$, $\{root, 3\}-t_3$, $\{root, 4\}-t_6$, and $\{root, 5\}-t_9$.

\subsubsection{Non-empty Neighboring Grids Query}
With a grid tree in place, we are now ready to describe the algorithm using the grid tree to identify all non-empty neighboring grids for each grid.

Given a grid $g_i$, to find its non-empty neighboring grids, the algorithm recursively traverses down the tree starting from the root node.
The algorithm first finds all child nodes of the root node with keys between $g_{i1}-\lceil \sqrt{d}\rceil$ and $g_{i1}+\lceil \sqrt{d}\rceil$.
If the root node has more than $MinPts$ child nodes, the item whose $\{t, key\}$ is equal to $\{root, g_{i1}\}$ in $h$ is used to identify the child node with the smallest key that meets the condition, and then iteratively call $NEXT$ to find all nodes with keys between $g_{i1}-\lceil \sqrt{d}\rceil$ and $ g_{i1}+\lceil \sqrt{d}\rceil$.
Otherwise, these nodes are found by examining each child node of the root node.
Denote these nodes as $\Phi$.
For each node $t\in \Phi$, the current $offset$ of $t$ is $t.offset = \max\{|KEY(t) - g_{i1}|-1,0\}^2$.
The offset of $t$ indicates that the minimum distance between $g_i$ and any grid in this subtree will not less than $\sqrt{t.offset}\cdot \epsilon/\sqrt{d}$.
Therefore, all non-empty neighboring grids of $g_i$ are in the subtrees of nodes in $\Phi$.
This is because, for each grid $g_j$ not in the subtrees of the nodes in $\Phi$, we have $g_{j1}-g_{i1}>\lceil \sqrt{d}\rceil$.
It follows that the minimum distance between $g_i$ and $g_j$ is greater than $\epsilon$ ($(\lceil \sqrt{d}\rceil+1)\cdot\epsilon/ \sqrt{d} > \epsilon$).
Therefore, there is no need to consider all nodes not in the subtrees of the nodes in $\Phi$.
As a result, a lot of redundant computations can be avoided.
This significantly improves the performance of our algorithm.

After that, for each node $t\in \Phi$, its child nodes with keys between $g_{i2}-\lceil \sqrt{d}\rceil$ and $g_{i2}+\lceil \sqrt{d}\rceil$ are found.
For each such node $t^\prime$, its offset is calculated by
\begin{equation}
\label{offset}
    t^\prime.offset = t.offset+\max\{|KEY(t) - g_{i2}|-1,0\}^2.
\end{equation}
All these nodes will be stored in the set $temp$.
Nodes in $temp$ with $offset$ greater than or equal to $d$ will be further excluded.
Then, $\Phi$ is updated to $temp$.
This procedure is repeated until all nodes in $\Phi$ are leaf nodes.

To sum up, we find all non-empty neighboring grids of $g_i$ by calling $NeighboringGridsQuery(T, g_i)$.
The procedure $NeighboringGridsQuery$ is summarized in Algorithm \ref{ngq}.
We denote the non-empty neighboring grids of $g_i$ as $Nei(g_i)$.

It is worth noting that the non-empty neighboring grids are sorted in ascending order using counting sort according to their offset.
This will significantly accelerate the performance of identifying core points.
For one thing,  a point's neighbors are more likely to be in the grid close to it.
For another, given a point, once we confirm that it has more than $MinPts$ neighbors, it will be identified as a core point without finding all its neighbors.
This saves a lot of unnecessary calculations.

\begin{algorithm}[ht]
	\caption{NeighboringGridsQuery}
	\label{ngq}
	\begin{algorithmic}[1]
        \REQUIRE $T$: a grid tree; $g_i$: a grid.
        \ENSURE $N$:  the non-empty neighboring grids of $g_i$.
        \STATE $T.root.offset = 0$
        \STATE $\Phi = \{T.root\}$
        \FOR {$j=1:d$}
        \STATE $temp = \emptyset$
        \FORALL {$nd\in \Phi$}
        \STATE $childNodes$ = all child nodes of $nd$ with keys between $g_{ij}-\lceil \sqrt{d}\rceil$ and $g_{ij}+\lceil \sqrt{d}\rceil$
        \FORALL {$child\in childNodes$}
        \STATE Calculate $child$'s offset by (\ref{offset}).
        \IF {$child.offset\geq d$}
        \STATE Remove $child$ from $childNodes$.
        \ENDIF
        \ENDFOR
        \STATE $temp = temp \cup childNodes$
        \ENDFOR
        \STATE $\Phi = temp$
        \ENDFOR
        \STATE $N=\emptyset$
        \FORALL {$nd\in \Phi$}
        \STATE $N= N \cup \{GRID(nd)$\}
        \ENDFOR
        \STATE Sort $N$ by radix sort according to their offsets.
        \RETURN $N$
	\end{algorithmic}
\end{algorithm}

For example, consider the query of $g_6 = (3,3)$ in Fig. \ref{gridTree}.
Start from the root node; the algorithm first finds its child nodes with keys between 1 and 5, and the result is $\Phi=\{t_3, t_6,t_9,t_{13}\}$.
Then for $t_3$, the algorithm first finds the child nodes of $t_3$ with keys between 1 and 5.
Clearly, both child nodes of $t_3$ meet the condition.
But $t_4.offset = 2$, so it is excluded.
Other nodes in $\Phi$ are processed in a similar way.
Finally, $\Phi$ is updated to $\{t_5, t_8,t_{11},t_{12},t_{14}\}$.
Therefore, the non-empty neighboring grids of $g_6$ are $g_3,g_5,g_7,g_8,g_9$.

\subsubsection{Complexity Analysis}
A grid tree has $d+1$ levels, and each level contains at most $|Gs|$ nodes.
In addition, the space complexity of the hash table is $O(d\cdot \sqrt{d}\cdot |Gs|)=O(d^{3/2}\cdot |Gs|)$.
So the overall space complexity of grid tree is $O(d^{3/2}\cdot |Gs|)$.
To build a grid tree, we only need to scan $Gs$ once and update the hash table.
Obviously, the expected complexity of updating the hash table is $O(d^{3/2}\cdot |Gs|)$.
Therefore, Algorithm \ref{bgt} runs in $O(d\cdot |Gs| + d^{3/2}\cdot |Gs|) = O(d^{3/2}\cdot |Gs|)$ expected time.

To find the non-empty neighboring grids of a given grid, at the $i$th level, there are at most $(2\lceil\sqrt{d}\rceil+1)^i$ nodes in $\Phi$.
Sorting $Nei(g_i)$ runs in $O(|Nei(g_i)|+d)$ because the offset of each node in $Nei(g_i)$ is less than $d$.
It follows that Algorithm \ref{ngq} runs in $O(d\cdot(2\lceil\sqrt{d}\rceil+1)^d)$ worst case time, which is constant for any fixed $d$.
On the other hand, if each node in $nodes$ contains at least one neighboring grid of $g_i$ in its subtree, then the running time will be $O(d\cdot|Nei(g_i)|)$.

\subsection{Merging Two Grids}\label{merging two grids}
Apart from finding non-empty neighboring grids, the merging step is also time-consuming as mentioned in Section \ref{preliminaries}.
In order to reduce the execution time of the merging step, we propose a new merging algorithm that can prune unnecessary distance calculations by making use of the spatial relationships among points.
We also show that the proposed merging algorithm is theoretically effective.

\subsubsection{Motivation}
Given two core grids $g_i$ and $g_j$, our goal is to determine whether they can be merged.
Let $s_i$ and $s_j$ be the sets of core points in $g_i$ and $g_j$, respectively.
Let $m_i$ and $m_j$ be the number of points in $s_i$ and $s_j$, respectively.
Clearly, $s_i$ and $s_j$ are linearly separable, that is, there exists a hyperplane such that all points in $s_i$ are on or below the hyperplane, and all points in $s_j$ are on or above the hyperplane.
Based on Definition \ref{merge}, $g_i$ and $g_j$ can be merged if and only if $MinDist(s_i, s_j)\leq \epsilon$.

\begin{definition}
	\label{md}
    The minimum distance between two sets $s_i$ and $s_j$ is defined as
    \begin{equation}
    MinDist(s_i, s_j) = \underset{p \in s_i, q \in s_j}{\min}dist(p,q).
    \end{equation}
\end{definition}
A straightforward way is to calculate the minimum distance between $s_i$ and $s_j$ using the brute force algorithm.
However, it takes $O(m_im_j)$ time to calculate the minimum distance between $s_i$ and $s_j$ with the brute force algorithm as it needs to calculate the distances between every pair of points $p\in s_i$ and $q \in s_j$.
This is unacceptable when both $m_i$ and $m_j$ are large.

Several algorithms have been proposed to reduce the time complexity of determining whether two grids can be merged.
G13 \cite{gunawan2013faster} runs in $O(n\log n)$ time using a Voronoi diagram to find the nearest point in a neighboring grid.
Although the complexity of $\rho$-approximate DBSCAN \cite{gan2017hardness} reduces down to $O(n)$ in expectation, it permits inaccuracy in the result.
In order to reduce unnecessary distance calculations, KNN-BLOCK DBSCAN \cite{chen2019knn} introduces a simple technique to eliminate impossible points in each set before using the brute force algorithm.
In KNN-BLOCK DBSCAN, there is a point in each set as the center of the set.
Let $c_i$ and $c_j$ be the centers of $s_i$ and $s_j$, respectively.
The simple technique eliminates point $p \in s_i$ if it satisfies $dist(p, c_j) > \epsilon + \xi_j$, where $\xi_j = \max_{q \in s_j}dist(c_j,q)$.
However, when the points of $s_i$ are all within $N_{\epsilon + \xi_j}(c_j)$ and all points of $s_j$ are within $N_{\epsilon + \xi_i}(c_i)$, the running time of checking whether $g_i$ and $g_j$ can be merged is $O(m_im_j)$.
In\cite{CHEN2021107624}, Chen et al. proposed the BLOCK-DBSCAN which uses an approximate merging algorithm.
It first selects a random point in $s_i$ as a seed, and identifies the closest point to the seed in another set as a new seed.
The merging algorithm repeats this procedure until convergence.
Although the merging algorithm in BLOCK-DBSCAN is efficient and converges quickly, the result may be inaccurate.

To summarise, there are few algorithms that guarantee both accuracy and low complexity.
Therefore, we propose a fast merging algorithm based on the spatial relationships between points to reduce unnecessary distance calculations which dominate the execution time of the merging step.

\subsubsection{Fast Merging Algorithm}
Recall that our goal is to check whether two core grids can be merged.
In this paper, we introduce a fast merging algorithm utilizing the spatial relationships among points to efficiently solve the problem with fewer required distance calculations.
We first introduce two pruning strategies founded on spatial relationships among points.
Then, we present the pruning method based on these two pruning strategies.
Finally, we introduce the merging algorithm in detail.

Without loss of generality, in Fig. \ref{basicIdea} we take two linearly separable sets as an example to illustrate the general ideas of these two pruning strategies since the new merging algorithm can be generalized to determine whether the minimum distance between any two linearly separable sets is less than a certain threshold.
Assume that we have already calculated the distance from point $p\in s_i$ to all points in $s_j$.
The closest point to $p$ is point $q\in s_j$, and the corresponding distance is $dist(p,q) > \epsilon$.
For convenience, let $\sigma = dist(p,q) - \epsilon$.
Based on the spatial relationships among points, we propose the following two pruning strategies used in our pruning method.

\begin{figure}[!t]
\centering
\includegraphics[width=2.5in]{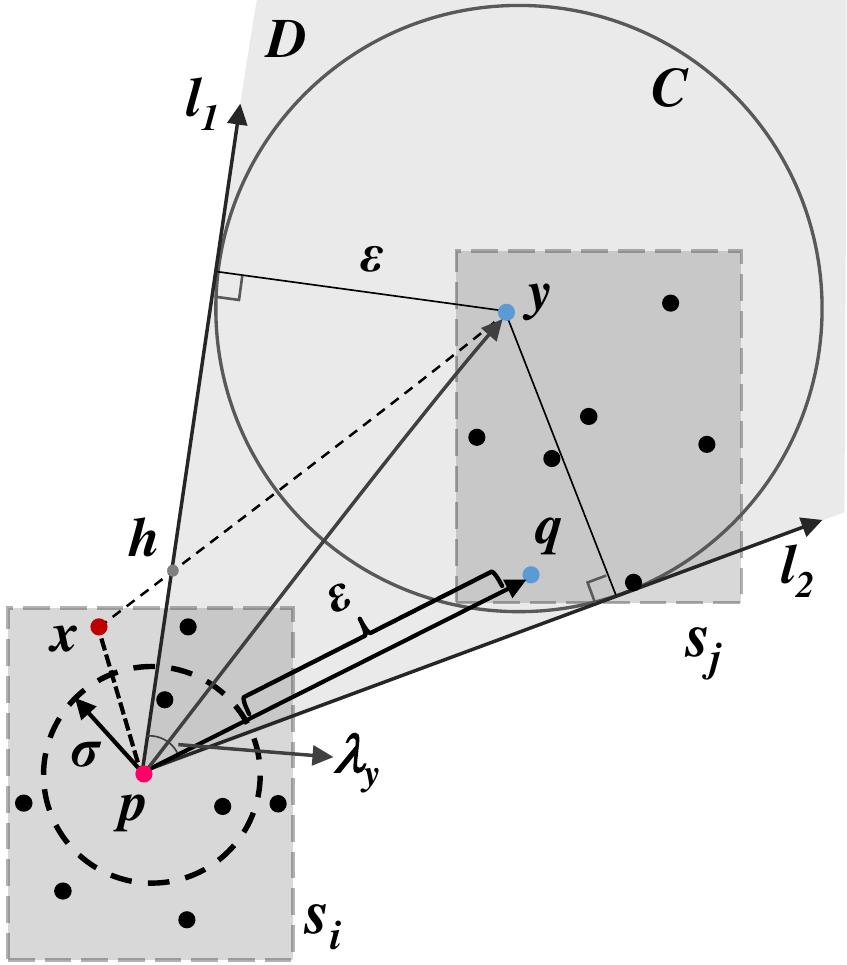}
\caption{The general ideas of our two pruning strategies.}
\label{basicIdea}
\end{figure}

In fact, if $p$ is far away from every point in $s_j$, the point close enough to $p$ will be far away from every point in $s_j$ as well.
This motivates us to propose the first pruning strategy, namely, triangle inequality pruning strategy, to prune trivial points.
Here, a point $x$ is a trivial point if the minimum distance between set $\{x\}$ and $s_j$ is greater than $\epsilon$.
For each point $x\in s_i$, if $dist(x, p)<\sigma$, then $x$ is a trivial point.
The reasoning is that for every point $y\in s_j$, the distance between $x$ and $y$ satisfies
\begin{IEEEeqnarray}{rCl}
\label{tri}
dist(x, y)&&{}\geq{}dist(p, y) - dist(x, p)\IEEEnonumber\\
&&{}> dist(p, q) - \sigma \IEEEnonumber\\
&&{}=\epsilon.
\end{IEEEeqnarray}
The first inequality is based on the triangle inequality. As a result, all points inside the $\sigma$-ball centered at $p$ can be pruned.

The second pruning strategy is angle-based pruning strategy.
We first introduce our findings, then explain in detail our angle-based pruning strategy.

As shown in Fig. \ref{basicIdea}, $C$ is a circle with radius $\epsilon$ and $y$ as its center.
$l_1$ and $l_2$ are the two tangent lines of circle $C$ passing through $p$. Evidently, the distance from $y$ to any point on the tangent line is greater than or equal to $\epsilon$. Let $D$ denote the region surrounded by the two tangent lines $l_1$ and $l_2$.
It is simple to verify that, for every point in $x\in s_i$, if it is outside $D$, the distance between $x$ and $y$ is greater than $\epsilon$.
The reasoning is that the line $xy$ must intersect one of the tangent lines, and the distance from the intersection point to $y$ is not less than $\epsilon$.

For example: in Fig. \ref{basicIdea}, the line segment $xy$ intersects $l_1$ at $h$, it can be concluded that $dist(x, y)$ is greater than $\epsilon$ since
\begin{equation}
dist(x, y) = dist(x, h) + dist(y, h)\nonumber
\end{equation}
where $dist(y, h) \geq \epsilon$ and $dist(x, h) > 0$.

For each point $y\in s_j$, the maximum angle of $y$ used to determine which points in $s_i$ are not in the $N_{\epsilon}(y)$ is defined as follows.

\begin{definition}
	\label{angle}
    Given $s_i$, $s_j$, and $p\in s_i$, the maximum angle of $y\in s_j$ wrt. $p$ is defined as
    \begin{equation}
    \label{lambda}
    \lambda_y = \arcsin \frac{\epsilon}{dist(p,y)} + \arccos \frac{\overrightarrow{p q}\cdot \overrightarrow{p y}}{dist(p,q)\times dist(p,y)}
    \end{equation}
    where $q = {\arg\min}_{y \in s_j}dist(p,y)$.
\end{definition}

The following theorem shows that for $\forall x\in s_i$, if the angle between $\overrightarrow{p q}$ and $\overrightarrow{p x}$ is greater than the maximum angle of $y$, then $x$ is not in the $\epsilon$-neighborhood of $y$.

\begin{theorem}
	\label{angle lemma}
    Given $s_i, s_j$, and a point $p\in s_i$ such that $\min_{z \in s_j}dist(p,z)>\epsilon$. Let $q = {\arg\min}_{z \in s_j}dist(p,z)$.
    For each $x\in s_i$, if the angle between $\overrightarrow{p q}$ and $\overrightarrow{p x}$ is greater than $\lambda_y$, then
    \begin{equation}
    dist(x, y) > \epsilon.
    \end{equation}
    Furthermore, if $s_i$ and $s_j$ are the sets of core points in $g_i$ and $g_j$, respectively, where $g_i\in Nei(g_j)$,
    then, for each point $y\in s_j$, the maximum angle of $y$ wrt. $p$ is less than $5\pi/6$.
\end{theorem}
\begin{proof}
For convenience, we use $\langle\overrightarrow{p q}, \overrightarrow{p x}\rangle$ to represent the angle between $\overrightarrow{p q}$ and $\overrightarrow{p x}$.

First, we prove that for every point $x\in s_i$ satisfies $dist(x,y)\leq \epsilon$, then
\begin{equation}
\label{thresh}
\theta_2 \leq \arcsin \frac{\epsilon}{dist(p,y)}
\end{equation}
where $\theta_2 = \langle\overrightarrow{p x}, \overrightarrow{p y}\rangle$ and $0\leq \theta_2\leq \pi$.
Let $C$ be the hypersphere with radius $\epsilon$ and $y$ as its center.
Denote by $l$ the straight line defined by $p$ and $x$.
Clearly, $l$ intersects with $C$. If $l$ intersects with $C$ at two points, then the distance from $y$ to $l$, denoted by $d_{l,y}$, will not greater than $\epsilon$ since the distance from the chord to the centre will be less than the radius.
If $l$ intersects with $C$ at one points, then $d_{l,y}=\epsilon$.
Therefore, the distance from $y$ to $l$ does not exceed $\epsilon$. So
\begin{IEEEeqnarray}{rCl}
\theta_2&&{}={}\arcsin \frac{d_{l,y}}{dist(p,y)}\IEEEnonumber\\
&&{}\leq {}\arcsin \frac{\epsilon}{dist(p,y)}.\IEEEnonumber
\end{IEEEeqnarray}

Second, we prove that if $\langle\overrightarrow{p q}, \overrightarrow{p x}\rangle > \lambda_y$, then $dist(x, y)$ is greater than $\epsilon$. Let $\theta_1 = \langle\overrightarrow{p q}, \overrightarrow{p y}\rangle$, $\gamma = \langle\overrightarrow{p q}, \overrightarrow{p x}\rangle$. Clearly, $0\leq \theta_1,\gamma\leq \pi$.
According to the spherical law of cosines\cite{todhunter1863spherical}, the following equation holds:
\begin{IEEEeqnarray}{rCl}
\cos{\gamma}{}={}&&\cos{\theta_1} \cos{\theta_2}+\sin{\theta_1} \sin{\theta_2} \cos{\omega}\IEEEnonumber
\end{IEEEeqnarray}
where $\omega$ is the angle between two intersecting planes. One of the planes is determined by $p,x,y$, and the surface of the unit sphere centered at $p$. The other plane is determined by $p,q,y$, and the surface of the unit sphere centered at $p$. Further, we have
\begin{IEEEeqnarray}{rCl}
\label{threeangle}
\cos{\gamma}&&{}\geq{}\cos{\theta_1} \cos{\theta_2}-\sin{\theta_1} \sin{\theta_2}\IEEEnonumber\\
&&{}={}\cos{(\theta_1 + \theta_2)}
\end{IEEEeqnarray}
since $\cos{\omega} \geq -1$ and $\sin{\theta_1} \sin{\theta_2} \geq 0$. Based on (\ref{threeangle}), we have
\begin{IEEEeqnarray}{rCl}
\label{theta2}
\theta_2&&{}\geq{}\gamma-\theta_1\IEEEnonumber\\
&&{}>{}\lambda_y - \theta_1\IEEEnonumber\\
&&{}={}\arcsin \frac{\epsilon}{dist(p,y)}.
\end{IEEEeqnarray}

According to the first step of the proof and (\ref{theta2}), we can conclude that $dist(x, y) > \epsilon$.

If $s_i$ and $s_j$ are the sets of core points in $g_i$ and $g_j$, respectively, where $g_i\in Nei(g_j)$.
Since $q$ and $y$ are in the same grid, hence $dist(q,y)\leq \epsilon$.
Let $a=dist(p,q),b=dist(p,y)$.
We have $\epsilon<a \leq b\leq 2\epsilon$ since $g_i\in Nei(g_j)$.
Then, we have
\begin{IEEEeqnarray}{rCl}
\lambda_y&&{}={}\arcsin \frac{\epsilon}{b} + \langle\overrightarrow{p q}, \overrightarrow{p y}\rangle\IEEEnonumber\\
&&{}<{}\arcsin 1 + \arccos \frac{a^2+b^2-dist(q,y)^2}{2ab}\IEEEnonumber\\
&&{}\leq{}\frac{\pi}{2} + \arccos \frac{a^2+b^2-\epsilon^2}{2ab}\IEEEnonumber\\
&&{}< \frac{\pi}{2} + \arccos \frac{1}{2}\IEEEnonumber\\
&&{}={}\frac{5\pi}{6}\IEEEnonumber
\end{IEEEeqnarray}
where the third inequality holds since $(a^2+b^2-\epsilon^2)/(2ab) > 1/2$.
This completes the proof.
\end{proof}

Based on the above findings, we propose the angle-based pruning strategy to further prune trivial points. Let
\begin{equation}
\lambda = \underset{y \in s_j}{\max}\lambda_y.
\end{equation}
Then, if the angle between $\overrightarrow{p q}$ and $\overrightarrow{p x}$ is greater than $\lambda$, we can conclude that $x$ is a trivial point without calculating the exact minimum distance between set $\{x\}$ and $s_j$.

Moreover, if $s_i$ and $s_j$ are the sets of core points in $g_i$ and $g_j$, respectively, then by combining Theorem \ref{angle lemma}, we have $\lambda<5\pi/6$.
It follows that there is always a non-empty area in the $g_i$ such that the points in this area are all trivial points. This allows us to further prune some trivial points in $s_i$ to reduce unnecessary distance calculations.

\begin{figure}[!t]
\centering
\includegraphics[width=2.5in]{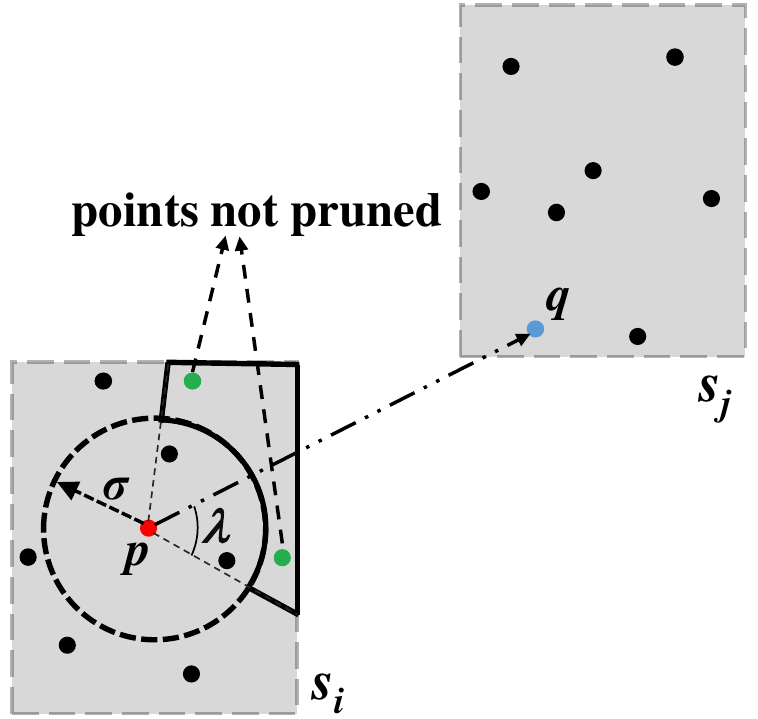}
\caption{An example of pruning trivial points in $s_i$.}
\label{boot}
\end{figure}

Based on the above two pruning strategies, we propose the pruning method which removes trivial points using the triangle inequality and the angle relationship among points.
For each point $x$ in $s_i$, $x$ is a trivial point if it lies in the $\sigma$-neighborhood of $p$ or the angle between $\overrightarrow{p q}$ and $\overrightarrow{p x}$ is greater than $\lambda$.
The pruning method based on spatial relationships, is summarized in Algorithm \ref{spa}.

\begin{algorithm}[ht]
	\caption{Pruning}
	\label{spa}
	\begin{algorithmic}[1]
        \REQUIRE point sets $s_i, s_j$; points $p, q$; $\epsilon$.
        \ENSURE Remove trivial points in $s_i$.
        \STATE $\lambda = 0$
        \FORALL{$y \in s_j$}
        \STATE Compute $\lambda_y$ by (\ref{lambda}).
        \IF {$\lambda_y > \lambda$}
        \STATE $\lambda = \lambda_y$
        \ENDIF
        \ENDFOR
		\FORALL{$x \in s_i$}
        \STATE compute $dist(x, p)$
        \IF {$dist(x, p) < dist(p, q) - \epsilon$}
        \STATE remove $x$ from $s_i$ /*Prune trivial points in $s_i$ using the triangle inequality pruning strategy.*/
        \ELSE
        \STATE $\theta = \arccos \frac{\overrightarrow{p q}\cdot \overrightarrow{p x}}{dist(p,q)\times dist(p,x)}$ /*Compute the angle between $\overrightarrow{p q}$ and $\overrightarrow{p x}$.*/
        \IF {$\lambda < \theta$}
        \STATE remove $x$ from $s_i$ /*Prune trivial points in $s_i$ using the angle-based pruning strategy.*/
        \ENDIF
        \ENDIF
        \ENDFOR
        \RETURN $s_i$
	\end{algorithmic}
\end{algorithm}

Consider the sets in Fig. \ref{basicIdea}.
The result of pruning trivial points in $s_i$ is illustrated in Fig. \ref{boot}.
Points inside the region surrounded by solid lines won't be pruned.
This is because the region is surrounded by vectors whose angle to $\overrightarrow{p q}$ equals $\lambda$ and an arc whose radius is $\sigma$ and center is $p$.
It follows that a point $x$ lying in this region is not a trivial point since $dist(p,x)\ge\sigma$ and the angle between $\overrightarrow{p q}$ and $\overrightarrow{p x}$ is less than $\lambda$.
Ultimately, seven of nine points in $s_i$ have been pruned, with three pruned by the angle-based pruning strategy and four pruned by the triangle inequality pruning strategy.
It is easier to determine whether $MinDist(s_i, s_j)$ is greater than $\epsilon$ when there are only two points in $s_i$.

We can now state the fast merging algorithm.
The fast merging algorithm uses the pruning method to remove trivial points when checking whether two grids can be merged, so as to reduce unnecessary distance calculations.
The pseudocode of the fast merging algorithm is summarized in Algorithm \ref{fma}.

\begin{algorithm}[ht]
	\caption{FastMerging}
	\label{fma}
	\begin{algorithmic}[1]
        \REQUIRE point sets $s_i, s_j$; $\epsilon$.
        \ENSURE If $MinDist(s_i, s_j)\leq \epsilon$, return yes. Otherwise, return no.
        \STATE Randomly selected a point from $s_i$, denoted by $p$.
        \REPEAT
        \STATE $q = \arg\underset{y \in s_j}{\min}dist(p,y)$
        \IF {$dist(p,q) \leq \epsilon$}
        \RETURN yes
        \ELSE
        \STATE /*Remove trivial points in $s_i$ by Algorithm \ref{spa}.*/
        \STATE $s_i = Pruning(s_i,s_j,p,q,\epsilon)$
        \ENDIF
        \STATE $p = \arg\underset{x \in s_i}{\min}dist(x,q)$
        \IF {$dist(p,q) \leq \epsilon$}
        \RETURN yes
        \ELSE
        \STATE /*Remove trivial points in $s_j$ by Algorithm \ref{spa}.*/
        \STATE  $s_j = Pruning(s_j, s_i,q,p,\epsilon)$
        \ENDIF
        \UNTIL{$|s_i|=0$ \textbf{or} $|s_j|=0$}
        \RETURN no
	\end{algorithmic}
\end{algorithm}

As Algorithm \ref{fma} shows, the fast merging algorithm consists of four steps:

\begin{enumerate}[\IEEEsetlabelwidth{12)}]
\item Let $p$ be a random point in $s_i$.
\item The algorithm finds $p$'s nearest point $q\in s_j$, and checks whether $dist(p, q) \leq \epsilon$.
If so, the algorithm returns yes. Otherwise, the algorithm removes all trivial points in $s_i$ by Algorithm \ref{spa}.
\item Update $p$ to $q$'s nearest point in $s_i$. If $dist(p, q) \leq \epsilon$, the algorithm returns yes.
Otherwise, all trivial points in $s_j$ will be removed using Algorithm \ref{spa}.
\item If one of the sets is empty, then return no. Otherwise, go back to step 2.
\end{enumerate}

Clearly, FastMerging is an iterative algorithm.
In each iteration, trivial points are removed, and the distance from these points will not be calculated in the subsequent steps.
This ensures that our algorithm runs efficiently.

\begin{theorem}
	\label{fma lemma}
    Algorithm \ref{fma} returns yes if and only if there are points $p\in s_i$ and $q\in s_j$ such that $dist(p, q)\leq \epsilon$.
\end{theorem}
\begin{proof}
Obviously, Algorithm \ref{fma} returns yes only if there are core points $p$ and $q$ satisfying $dist(p,q) \leq \epsilon$.
Therefore, if Algorithm \ref{fma} returns yes, grids $g_i$ and $g_j$ are density reachable from each other.

On the contrary, if Algorithm \ref{fma} returns no, then there are no $p$ and $q$ such that $dist(p,q)\leq \epsilon$.
This is because Algorithm \ref{fma} returns no only when one of the sets is empty, and recall that only the trivial points will be removed.
\end{proof}

\subsubsection{Complexity Analysis}\label{merging complexity}
Obviously, the complexity of the fast merging algorithm depends on two factors: (\romannumeral1) the number of distance calculations in each iteration, and (\romannumeral2) how many iterations it takes to terminate.

First, we show that the algorithm calculates the distance at most $m_i+m_j$ times in each iteration.
Let $r_{1,i}$ and $r_{2,i}$ be the number of points in $s_i$ and $s_j$ at the beginning of $i$th iteration respectively.
In each iteration, the algorithm calculates the distance $r_{1,i}+3r_{2,i}+2r_{1,i+1}$ times.
Then, we can conclude that, in each iteration, the algorithm calculates the distance at most $3(m_i+m_j)$ times since at least one point is removed from each set in each iteration.

Then, the complexity of the fast merging algorithm satisfies
\begin{IEEEeqnarray}{rCl}
O(\sum_{i=1}^{\kappa_{g_ig_j}}(r_{1,i}+3r_{2,i}+2r_{1,i+1})&&{}<{}O(\sum_{i=1}^{\kappa_{g_ig_j}}(m_i + m_j))\IEEEnonumber\\
&&{}={}O(\kappa_{g_ig_j}(m_i + m_j))\IEEEnonumber
\end{IEEEeqnarray}
where $\kappa_{g_ig_j}$ is the number of iterations.

Second, the following theorem gives the upper bound on the number of iterations.

\begin{theorem}
\label{com fma}
Given two core grids $g_i, g_j$, where $g_i\in Nei(g_j)$, let $s_i$ and $s_j$ be the sets of core points in $g_i$ and $g_j$, respectively.
Let $s_i^{\prime}$ denote the points in $s_i$ whose distances from each point in $s_j$ are all greater than $\epsilon$.
If $s_i^{\prime}$ is empty, then $\kappa_{g_ig_j}=1$. Otherwise, let
\begin{equation}
\tau = \underset{x \in s_i^{\prime}}{\min}\underset{y \in s_j}{\min}dist(x,y)-\epsilon.\nonumber
\end{equation}
Then $\kappa_{g_ig_j}$ is at most $V_g/V_{\tau/2}$, where $V_g$ is the volume of a grid which extends $\tau$ units to each side and $V_{\tau/2}$ is the volume of a ball with radius $\tau/2$.
\end{theorem}
\begin{proof}
If $s_i^{\prime}$ is empty, then in the 1th iteration $p\notin s_i^{\prime}$.
Based on the definition of $s_i^{\prime}$, there is at least one point $q\in s_j$ such that $dist(p,q)\leq \epsilon$.
Therefore, the algorithm returns yes in the 1th iteration, and $\kappa_{g_ig_j}=1$ holds.

On the contrary, $s_i^{\prime}$ is non-empty. If $p\in s_i^{\prime}$, all points inside the $\sigma$-ball centered at $p$ or whose angle exceeds $\lambda$ will be removed.
By the definition of $\tau$, we have $\sigma \geq \tau$.
As a result, in each iteration, if $dist(p,q)>\epsilon$, at least the points inside the $\tau$-ball centered at $p$ are removed.
Assume that in each iteration, only the points inside the $\tau$-ball centered at $p$ are removed. Then the number of iterations $\kappa_{g_ig_j}$ is not greater than the number of leaders in $s_i^\prime$ using $\tau$ as parameter.
According to the analysis of \cite{viswanath2009rough}, we have $\kappa_{g_ig_j}\leq V_g/V_{\tau/2}$.
This completes the proof.

\end{proof}

\begin{remark}
It should be noted that the upper bound on $\kappa_{g_ig_j}$ is exponential to the dimension $d$.
Nevertheless, there are two reasons that $\kappa_{g_ig_j}$ will not meet the upper bound.
First, it should be noted that if $s_i\backslash s_i^{\prime}$ is non-empty, the algorithm will return yes quickly.
The reason is that $q$ ($p$) is updated to the nearest point to $p$ ($q$) in $s_j$ ($s_i$), so the algorithm will obtain local optimal value quickly and return yes immediately if the local optimal value is not greater than $\epsilon$.
Second, in each iteration, the volume of the non-empty area from which the points will be removed is greater than the volume of a ball with radius $\tau$, so the algorithm will terminate faster than the version described in the proof of Theorem \ref{com fma} which only removes the points inside the $\tau$-ball centered at $p$.
\end{remark}

\begin{remark}\label{numerical error}
If slight inaccuracy in the results is acceptable, then there is an approximate FastMerging algorithm which runs in $O(1)$, regardless of the value of $\epsilon$ and the dimension $d$.
For an arbitrary constant $\delta>0$.
The approximate FastMerging algorithm differs from FastMerging only in line 4 and line 10.
In approximate FastMerging, it returns yes if $dist(p,q)\leq \epsilon+\delta$.
It follows that $\tau> \delta$ and the upper bound on the number of iterations satisfies $V_g/V_{\tau/2}< V_g/V_{\delta/2} = O(1)$, regardless of the value of $\epsilon$ and the dimension $d$.
\end{remark}

In summary, the complexity of the fast merging algorithm is linear to $m_i+m_j$.

\subsection{Overall Algorithm}\label{overall algorithm}

\begin{algorithm}[ht]
	\caption{GriT-DBSCAN}
	\label{FASTERDBSCAN}
	\begin{algorithmic}[1]
        \REQUIRE point sets $P$; parameters: $MinPts, \epsilon$.
        \ENSURE $C$: the clustering result.
        \STATE /*step 1: partitioning*/
        \STATE $Gs = Partitioning(P,\epsilon$)
        \STATE $T = BuildingGridTree(Gs)$.
        \STATE For each non-empty grid $g_i\in Gs$, use Algorithm \ref{ngq} to find its non-empty neighboring grids $Nei(g_i)$.
        \STATE /*step 2: identify core points*/
        \STATE Identify all core points in the data set like G13.
        \STATE /*step 3: merging*/
        \STATE Each core grid is marked as $unclassified$.
        \FORALL {core grid $g\in Gs$}
            \IF {$g$ is $unclassified$}
            \STATE Mark $g$ as $classified$.
            \STATE $seeds = \{g\}$
            \STATE $pos = 1$
                \WHILE {$pos \leq seeds.size()$}
                    \STATE $cur = seeds[pos]$
                    \FORALL {$unclassified$ core grid $g^\prime \in Nei(cur)$}
                        \STATE $s$ and $s^\prime$ are the sets of core points in $cur$ and $g^\prime$, respectively.
                        \IF {$FastMerging(s,s^\prime)=yes$}
                        \STATE Mark $g^\prime$ as $classified$.
                        \STATE $seeds = seeds \cup \{g^\prime\}$
                        \ENDIF
                    \ENDFOR
                \ENDWHILE
                \STATE All grids in $seeds$ form a cluster.
            \ENDIF
        \ENDFOR
        \STATE /*step 4: assign non-core points*/
        \STATE For each non-core point, check whether it is a noise or a border point.
        \RETURN $C$
	\end{algorithmic}
\end{algorithm}

Now, we can introduce the GriT-DBSCAN algorithm with complexity linear to the data set size by combining the above techniques.
Our algorithm consists of four steps as below.

First, the data set $P$ is partitioned using Algorithm \ref{partition}.
All non-empty grids are organized in a grid tree constructed by Algorithm \ref{gridTree}.
For each grid, its non-empty neighboring grids are found by Algorithm \ref{ngq} and stored in a vector.

Second, all core points in the data set are identified like G13.

Third, merge core grids to form clusters. In this step, we use Algorithm \ref{fma} to check whether two core grids can be merged.

At last, non-core points are identified as border points or noise points.
A non-core point is identified as a border point if and only if there is at least one core point in its $\epsilon$-neighborhood.
A non-core point is a noise point if it is not a border point.

The exact DBSCAN algorithm with complexity linear to the data set size is summarized in Algorithm \ref{FASTERDBSCAN}.

For the correctness and the time complexity of the GriT-DBSCAN algorithm, we present the following theorem.

\begin{theorem}
\label{correct}
The GriT-DBSCAN algorithm runs in $O(\kappa n + \eta)$ expected time, regardless of the value of $MinPts$ and the dimension $d$, where $\kappa=\max\{\kappa_{g_i,g_j}|g_i,g_j\in Gs, g_j \in Nei(g_i)\}$ is the maximum number of iterations in the merging step.
In addition, the clustering result of GriT-DBSCAN is consistent with the result of DBSCAN.
\end{theorem}
\begin{proof}
First, we prove the correctness of our algorithm, that is, the clustering result of GriT-DBSCAN is consistent with the result of DBSCAN.
Let $c$ be an arbitrary cluster found by the GriT-DBSCAN.
\begin{enumerate}[\IEEEsetlabelwidth{12)}]
\item For every core point $p\in c$ and any point $q\in P$ such that $q$ is density-reachable from $p$ wrt. $\varepsilon,MinPts$.
By the definition of density-reachable, there is a sequence of points $p_1,p_2,...,p_o\in P$ such that $p_1 = p$, $p_o = q$, and $p_{i+1}$ is directly density-reachable from $p_i$ for each $1\leq i\leq o-1$.
Furthermore, $p_1,p_2,...,p_{o-1}$ are core points and $dist(p_i,p_{i+1})\leq \varepsilon$ for $1\leq i\leq o-1$.
Denote by $g_{(i)}$ the grid in which $p_i$ lies, where $i=1,2,...,o$.
It follows from the definition of $G$ that $g_{(i)}$ and $g_{(i+1)}$ must be in the same connected component for $1\leq i\leq o-2$.
Hence, $p_1,p_2,...,p_{o-1}$ are in $c$.
If $p_o$ is a non-core point, then in the last step of GriT-DBSCAN, $p_o$ will also be assigned to $c$.
Otherwise, $g_{(o)}$ and $g_{(o-1)}$ must be in the same connected component of $G$.
It thus follows that $p_o\in c$.
Therefore, $p$ and $q$ are in the same cluster of GriT-DBSCAN.
\item Let $p$ and $q$ be two arbitrary points in $c$. We will show that $p$ is density-connected to $q$.
    If $p$ is a non-core point, based on the last step of GriT-DBSCAN, there is a core point $p^\prime\in c$ such that $p$ is directly density-reachable from $p^\prime$.
    Otherwise, we set $p^\prime = p$.
    Similarly, let $q^\prime$ be a core point in $c$ such that $q$ is directly density-reachable from $q^\prime$.
    Denote by $g_{p^\prime}$ and $g_{q^\prime}$ the grids covering $p^\prime$ and $q^\prime$, respectively.
    Since $p^\prime,q^\prime\in c$, $g_{p^\prime}$ and $g_{q^\prime}$ must be in the same connected component of $G$.
    Then, there is a sequence $g_{(1)},g_{(2)},...,g_{(o)}$ such that $g_{(1)}= g_{p^\prime}, g_{(o)}=g_{q^\prime}$.
    For each $i\in[1,o-1]$, $g_{(i)}$ and $g_{(i+1)}$ can be merged, which means that there are two core points $p^2_{i}\in g_{(i)}$ and $p^1_{i+1}\in g_{(i+1)}$ satisfying $dist(p^2_{i},p^1_{i+1})\leq \varepsilon$.
    Moreover, we have $dist(p^1_{i},p^2_{i})\leq \varepsilon$ since $p^1_{i},p^2_{i}$ in the same grid.
    Therefore, $p^\prime$ and $q^\prime$ are density-reachable from each other wrt. $\varepsilon, MinPts$.
    It follows that $p$ is density-connected to $q$ wrt. $\varepsilon, MinPts$.
\end{enumerate}
This completes the proof of the correctness.

Next, we will show that the GriT-DBSCAN algorithm runs in $O(\kappa n + \eta)$.
According to the analysis of \cite{gunawan2013faster}, the second step and the fourth step of Algorithm \ref{FASTERDBSCAN} runs in $O(n\cdot MinPts)$.
It takes $O(n+\eta)$ to partition the feature space.
For any fixed $d$, the expected time complexity of building the grid tree is $O(|Gs|)$. And the worst case time complexity of finding non-empty neighboring grids is $O(d\cdot(2\lceil\sqrt{d}\rceil+1)^d) = O(1)$.
It follows that the expected complexity of finding the non-empty neighboring grids for all grids is $O(|Gs|)$.
So the first step runs in $O(n)$ expected time since $|Gs| \leq n$.
In the third step, for any two core grids $g_i, g_j$, where $g_j\in Nei(g_i)$, we need to check whether they can be merged.
It follows from the analysis of Section \ref{merging complexity} that the time complexity of the third step is less than
\begin{IEEEeqnarray}{rCl}
\sum_{\substack{g_i\in Gs\\ g_j \in Nei(g_i)}}O(\kappa (m_i + m_j))&&{}={}\sum_{\substack{g_i\in Gs\\ g_j \in Nei(g_i)}}O(2 \kappa m_i) \IEEEnonumber\\
&&{}={}\sum_{g_i\in Gs}\sum_{\substack{g_j \in Nei(g_i)}}O(\kappa m_i) \IEEEnonumber\\
&&{}={}\sum_{g_i\in Gs}O(\kappa m_i) \IEEEnonumber\\
&&{}\leq {}O(\kappa n)\IEEEnonumber
\end{IEEEeqnarray}
where $m_i$ and $m_j$ are the number of core points in $g_i$ and $g_j$, respectively.
And the third equality is based on the fact that $g_i$ has $O(1)$ neighboring grids as long as $d$ is fixed.

In summary, the GriT-DBSCAN algorithm runs in $O(n+\eta + \kappa n + MinPts\cdot n) = O(\kappa n + \eta)$ expected time, regardless of the value of $MinPts$ and the dimension $d$.
This completes the proof.
\end{proof}

\begin{remark}
It is worth noting that the constant in the expected complexity of GriT-DBSCAN is at the order of $d\cdot(2\lceil\sqrt{d}\rceil+1)^d$, which is exponential to $d$.
Therefore, our algorithm is only suitable for low-dimensional data.
Moreover, in the experiments, $\kappa\leq 11$ is much smaller than the number of data points.
\end{remark}


\begin{remark}
Here we present our theoretical findings.
First, there is an exact DBSCAN algorithm that runs in $O(\kappa n)$ expected time, regardless of the value of $MinPts$ and the dimension $d$.
The algorithm differs from $\rho$-approximate DBSCAN only in the merging step.
We use the third step in GriT-DBSCAN to merge core grids. Following the analysis of \cite{gan2017hardness} and Theorem \ref{correct}, the algorithm runs in $O(\kappa n)$ expected time and the clustering result is consistent with the result of DBSCAN.
Second, if small inaccuracy in the clustering results are acceptable, this algorithm can be further modified to be an approximate DBSCAN algorithm running in $O(n)$ expected time according to Remark \ref{numerical error}, regardless of the value of $\epsilon$, $MinPts$, and the dimension $d$.
\end{remark}

\section{Experiments} \label{experiments}
Extensive experiments are conducted to evaluate GriT-DBSCAN and its two variants defined in Section \ref{exp1} by comparing with existing algorithms.
All the experiments are implemented on a machine equipped with a 2.5GHz CPU and 16GB memory using C++.
The operation system is Windows 10 64-bit.

\subsection{Data Sets and Parameter Settings}
In order to investigate the performance of the proposed algorithms, we conducted experiments on synthetic and real-world data sets.
In all data sets, we normalize each column to the integer domain of $[0,10^5]$.

The synthetic data sets are generated using the seed spreader (SS) generator produced by Gan and Tao \cite{gan2015dbscan, gan2017hardness}.
The seed spreader maintains a location when generating a synthetic data set.
It generates points uniformly in the neighborhood of the current location, and jumps to a random location with a certain probability.
In addition, seed spreader can generate data sets with either similar or variable density clusters.
We denote data sets with similar density clusters and variable density clusters as SS-simden and SS-varden, respectively.
Each synthetic data set contains 2 million points (unless specified otherwise), and $d\in \{2,3,5,7\}$.

We also use three real-world data sets to evaluate the performance of our algorithm.
PAM4D is a 4-dimensional data set with $3,850,505$ points, obtained by taking the first 4 principle components after performing PCA on the PAMAP2 data set \cite{uci}.
Farm is a 5-dimensional data set with $3,627,086$ points containing the VZ-features \cite{varma2003texture} of a satellite image of a farm in Saudi Arabia\footnote{http://www.satimagingcorp.com/gallery/ikonos/ikonos-tadco-farms-saudi-arabia.}.
House is a 7-dimensional data set with $2,049,280$ points obtained from UCI \cite{uci}, excluding date and time information.

For the synthetic data sets, we select the default values of $\epsilon$ and $MinPts$ to be those that produce the correct clustering results.
The default parameters for real-world data sets are similar to those found by Gan and Tao \cite{gan2017hardness}.
Unless specified otherwise, we use $\rho=0.01$ in the appr-DBSCAN algorithm.
\subsection{Experiments for $d\geq 3$} \label{exp1}
For $d\geq 3$, we compare the performance of the following algorithms:
\begin{itemize}
\item gan-DBSCAN \cite{gan2017hardness} is a grid-based exact DBSCAN for $d\geq 3$.

\item appr-DBSCAN \cite{gan2017hardness} is the state-of-the-art grid-based approximation DBSCAN reviewed in Section \ref{appr-dbscan}.

\item BLOCK-DBSCAN \cite{CHEN2021107624} is an approximation DBSCAN algorithm.
It first uses $\epsilon/2$-norm ball to identify inner core blocks, outer core points, and border points.
Second, it merges density-reachable inner core blocks into one cluster by an approximation algorithm.
Then, each outer core point is merged into a cluster which it is density-reachable.
Finally, border points will be assigned to corresponding clusters.

\item BLOCK-DBSCAN-FM is a variant of our algorithm obtained by combining FastMerging with BLOCK-DBSCAN.
BLOCK-DBSCAN-FM differs from BLOCK-DBSCAN only in the second step.
BLOCK-DBSCAN-FM utilizes cover tree \cite{beygelzimer2006cover} to index touch inner core points for fast $2\epsilon$-neighborhood queries.
Furthermore, our FastMerging algorithm developed in Section \ref{merging two grids} is used to merge density-reachable inner core blocks.
Based on Theorem \ref{fma lemma}, BLOCK-DBSCAN-FM is an exact DBSCAN algorithm.

\item GriT-DBSCAN is the exact DBSCAN algorithm we proposed in Section \ref{overall algorithm}.

\item GriT-DBSCAN-LDF is another variant of our algorithm by incorporating heuristics.
GriT-DBSCAN-LDF is an exact DBSCAN algorithm that differs from GriT-DBSCAN only in the merging step.
In GriT-DBSCAN-LDF, core grids are organized using union-find data structure \cite{tarjan1979class}.
In addition, the core grids are sorted in ascending order according to the number of core points using radix sort.
Then, these core girds are traversed in ascending order.
Fix a core grid $g_i$.
For each core grid $g_j\in Nei(g_i)$, we first check whether they are in the same set in the union-find data structure.
If so, we do nothing. Otherwise, we perform a "union" operation if they can be merged.
Finally, the core grids that belong to the same set in the union-find data structure belong to the same cluster.
(The reason why these core girds are traversed in ascending order is that we first perform merging checks on low-density core grids such that the cluster established soon. Consequently, high-density core grids can skip the merging checks since the cluster is already established \cite{boonchoo2019grid}. This reduces redundant merging checks.)
\end{itemize}

For gan-DBSCAN and appr-DBSCAN, we use the binary code which is written in C++ \cite{binaryCode}.
For BLOCK-DBSCAN, we use the source code written in C++ and publicly available \cite{block}.

\begin{figure*}[!ht]
\centering
\includegraphics[width=6.5in]{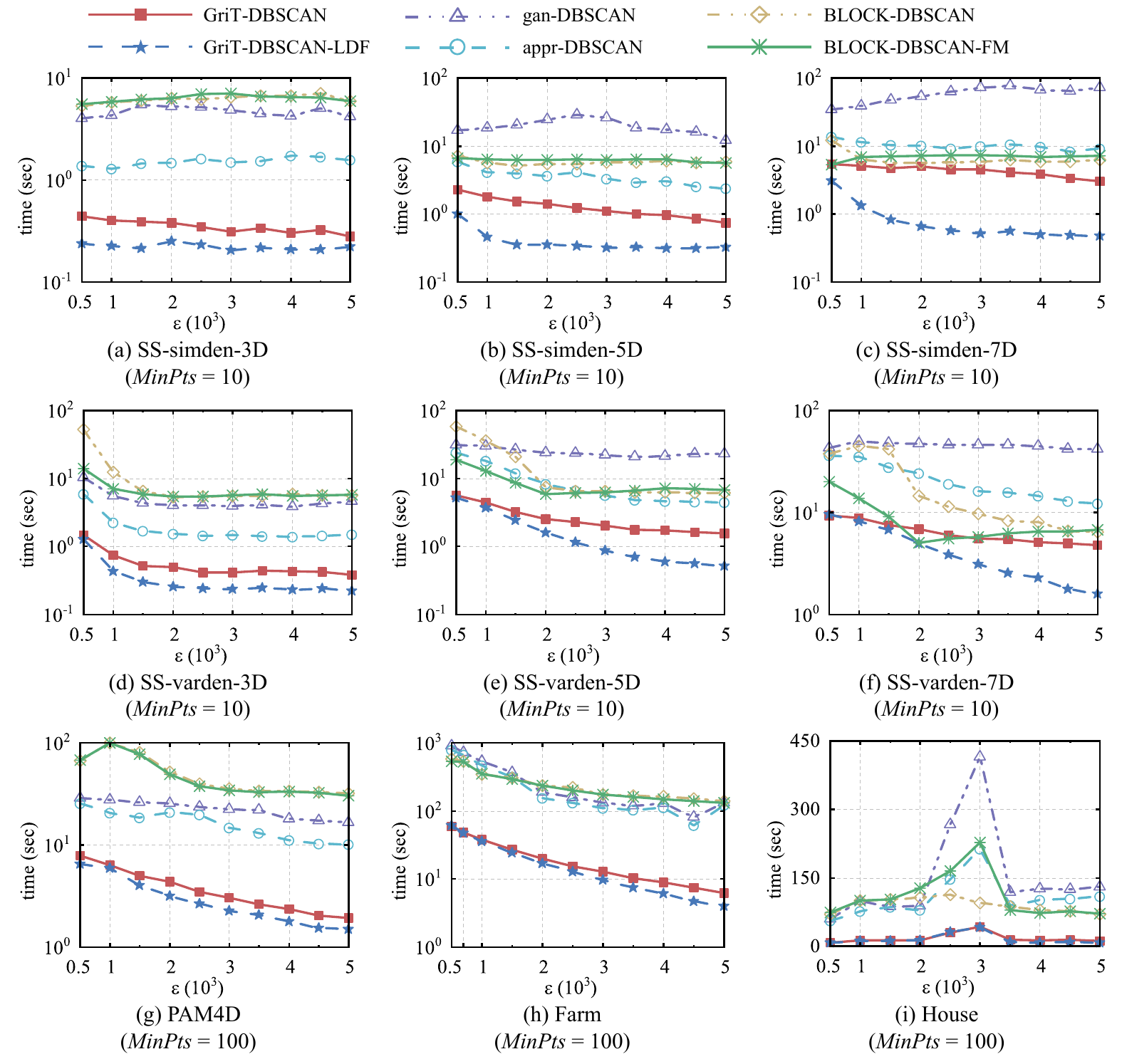}
\caption{Running time vs. $\epsilon$.}
\label{epsilon}
\end{figure*}

\begin{figure*}[!t]
\centering
\includegraphics[width=6.5in]{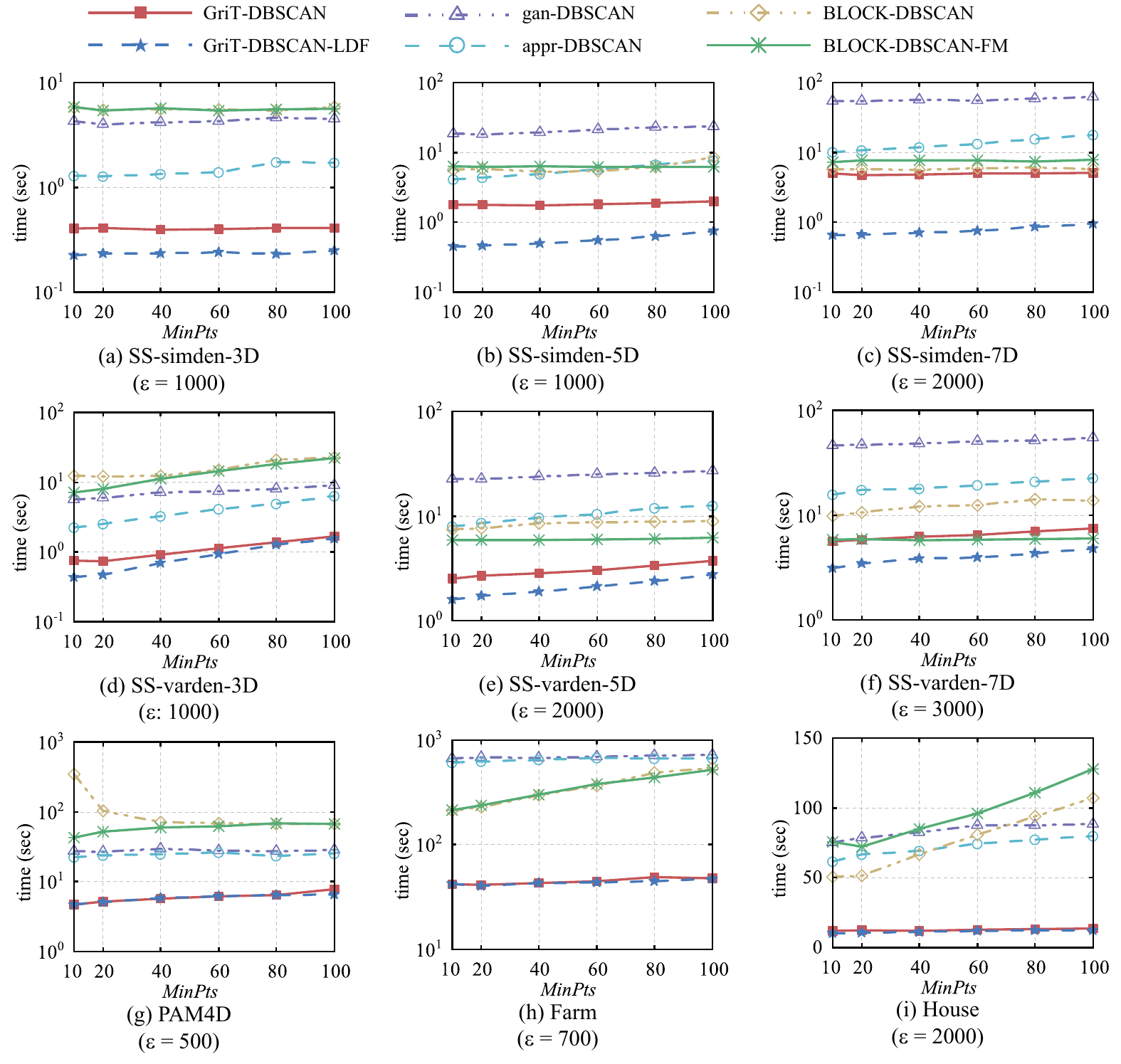}
\caption{Running time vs. $MinPts$.}
\label{minpts}
\end{figure*}

\begin{figure*}[!t]
\centering
\includegraphics[width=6.5in]{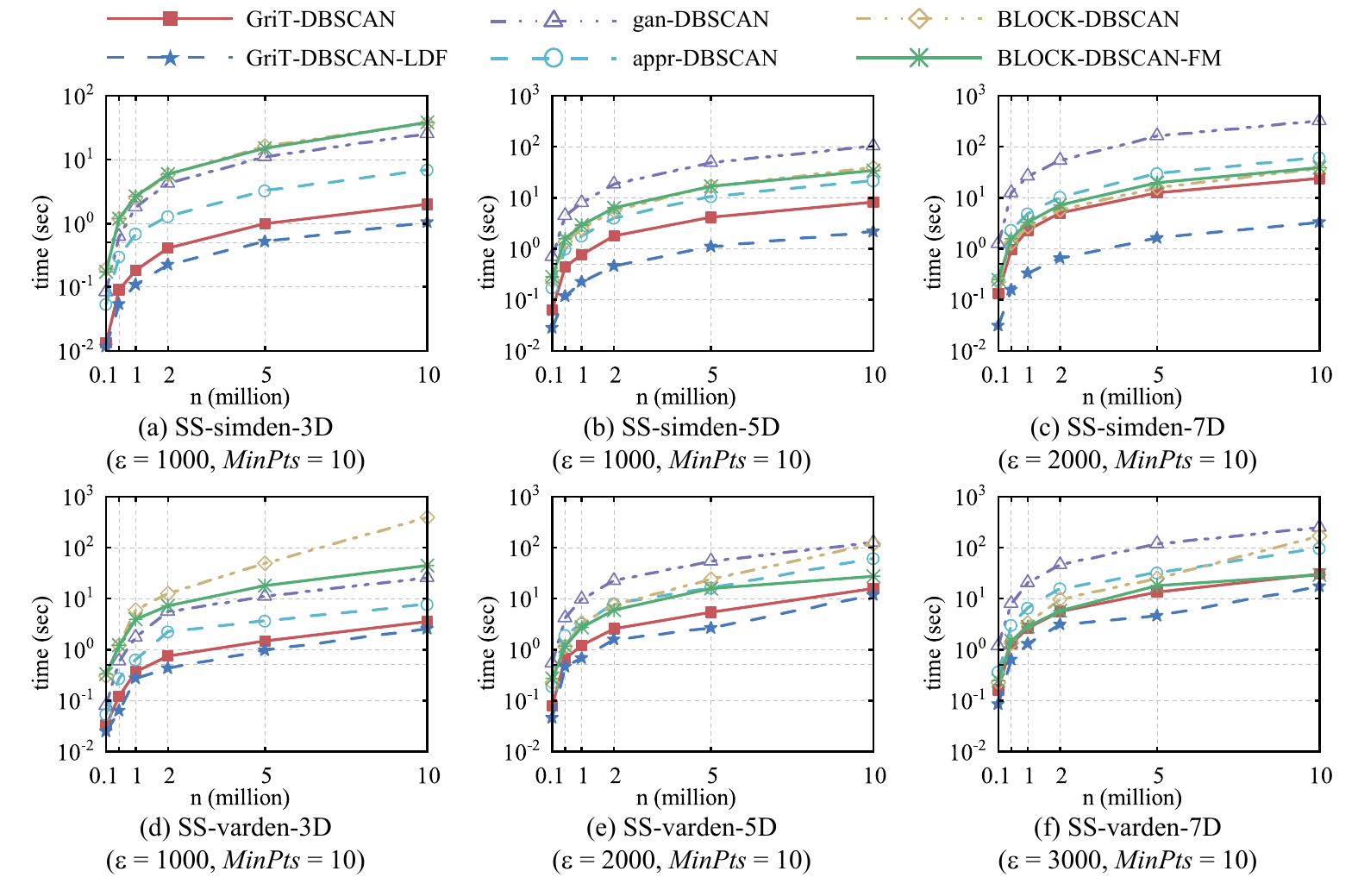}
\caption{Running time vs. $n$.}
\label{n}
\end{figure*}
\noindent\textbf{Influence of $\epsilon$.}
The first set of experiments aim to determine how $\epsilon$ affects the running time of each algorithm.
We fix $MinPts$ to the default value corresponding to the correct clustering result, and vary the parameter $\epsilon$ from 500 to 5000.
Fig. \ref{epsilon} shows the running time of GriT-DBSCAN, GriT-DBSCAN-LDF, gan-DBSCAN, appr-DBSCAN, BLOCK-DBSCAN, and BLOCK-DBSCAN-FM with different $\epsilon$.

For BLOCK-DBSCAN and BLOCK-DBSCAN-FM, as discussed in \cite{CHEN2021107624}, there is a trade off between the number of points filtered and the cost of range query: a large $\epsilon$ can filter large number of points, while the cost of range query get more expensive with $\epsilon$ increasing.
From Fig. \ref{epsilon}, the first effect dominates in most cases.
Therefore, the running time of BLOCK-DBSCAN and BLOCK-DBSCAN-FM tended to decrease as $\epsilon$ increases.
Particularly, BLOCK-DBSCAN-FM, with our FastMerging algorithm for fast merging check and the cover tree to index touch inner core points, obtains the accurate clustering results with nearly the same running time as BLOCK-DBSCAN or even less.
BLOCK-DBSCAN-FM's superiority to BLOCK-DBSCAN primarily depended on two factors.
The first factor is whether $\epsilon$ is small or large.
BLOCK-DBSCAN becomes slower as $\epsilon$ decreases because there are more touch inner core point, which increases the running time of using linear search to find the $2\epsilon$-neighborhood of each touch inner core point.
On the contrary, BLOCK-DBSCAN-FM benefits from applying a cover tree index, thus small $\epsilon$ has limited impact on its runtime.
Another factor is whether the density of clusters are similar or variable.
For a data set with variable density, the number of inner core blocks will be relatively large for the dense clusters; this makes BLOCK-DBSCAN slower.
For the SS-varden-7D data set as an example, with $\epsilon = 2000$, BLOCK-DBSCAN took 14.4834 seconds and discovered 5096 clusters.
BLOCK-DBSCAN-FM took 5.0483 seconds to produce 1563 clusters like other exact DBSCAN algorithms, which is around 200\% faster than BLOCK-DBSCAN.
However, the superiority disappeared when $\epsilon = 4500$; BLOCK-DBSCAN took 6.6477 seconds and discovered 1117 clusters, while
BLOCK-DBSCAN-FM took 6.5353 seconds to produce 15 clusters.
The superiority also disappeared in data set with similar density clusters.
For the SS-simden-7D data set, with $\epsilon = 2000$, BLOCK-DBSCAN took 5.7127 seconds and discovered 166 clusters, while BLOCK-DBSCAN-FM took 7.2591 seconds to produce 15 clusters.

For GriT-DBSCAN, GriT-DBSCAN-LDF, gan-DBSCAN, and appr-DBSCAN, they tend to perform better as $\epsilon$ increases since there are fewer non-empty grids, which speeds up the merging step.
It can be observed that GriT-DBSCAN and GriT-DBSCAN-LDF outperform gan-DBSCAN, appr-DBSCAN, and BLOCK-DBSCAN on all data sets.
In particular, GriT-DBSCAN-LDF benefits from the union-find data structure and the low density first traverse strategy for reducing redundant merging checks, significantly outperforming GriT-DBSCAN in most cases.
However, the improvement is negligible when $\epsilon$ is small.
The reason is that the number of non-empty grid is large when $\epsilon$ is small and most non-empty grids have one or two points, making the low density first traverse strategy ineffective.
Alternatively, GriT-DBSCAN-LDF is more competitive with big $\epsilon$.

In addition, gan-DBSCAN and appr-DBSCAN outperform BLOCK-DBSCAN and BLOCK-DBSCAN-FM when $d\leq 4$, but the improvement diminishes as $d$ increases.
This shows that BLOCK-DBSCAN and BLOCK-DBSCAN-FM are more scalable with dimension compared to gan-DBSCAN and appr-DBSCAN.
The reason is that the number of possible neighboring grids of gan-DBSCAN and appr-DBSCAN grows exponentially with $d$, this makes the cost of finding neighboring grids expensive.
In contrast, GriT-DBSCAN and GriT-DBSCAN-LDF use the grid tree to find the non-empty neighboring grids, with complexity $O(d\cdot (\sum^{|Gs|}_{i = 1}|Nei(g_i)|)/|Gs|)$ in the best case, which alleviates the influence of $d$ on the running time to some extent.

\noindent\textbf{Influence of $MinPts$.}
The next set of experiments aims to inspect how $MinPts$ influences the running time of each algorithm.
Therefore, we fix $\epsilon$ of each data set to the default value and vary $MinPts$ from 10 to 100.
Fig. \ref{minpts} shows the effects of $MinPts$ on the performance of different algorithms.
In general, GriT-DBSCAN, GriT-DBSCAN-LDF, gan-DBSCAN, and appr-DBSCAN become slower as $MinPts$ increases.
This is because the running time to identify all core points is $O(n\cdot MinPts)$.
However, as illustrated in Fig. \ref{minpts}, the influence of $MinPts$ was limited.
For BLOCK-DBSCAN and BLOCK-DBSCAN-FM, there is a trade off between the number of points filtered and the number of inner core blocks: a large $MinPts$ filter fewer points, while the number of inner core blocks decreases with increasing $MinPts$ leading to a reduction in the cost of merging density-reachable inner core blocks.
From Fig. \ref{minpts}, the first effect dominates except PAM4D and House.
Apparently, GriT-DBSCAN and GriT-DBSCAN-LDF outperform gan-DBSCAN, appr-DBSCAN, BLOCK-DBSCAN, and BLOCK-DBSCAN-FM on all data sets.

\noindent\textbf{Scalability with $n$.}
In the last set of experiments, we investigate the scalability of each algorithm with $n$ using the synthetic data sets.
To this end, we vary the number of points from 100k to 10m.
Other parameters are given their default values. The resulting running times are presented against $n$ in Fig. \ref{n}.
It can be seen that GriT-DBSCAN-LDF consistently outperforms all other algorithms.
Once again, BLOCK-DBSCAN and BLOCK-DBSCAN-FM become competitive as $d$ increases, and BLOCK-DBSCAN is considerably slower than BLOCK-DBSCAN-FM on data sets with variable density clusters, which confirm our analysis in the first set of experiments.

\subsection{Experiments for $d = 2$}
Next, we perform a set of experiments to compare the efficiency of our algorithms with existing algorithms in $d=2$.
In particular, we compare the performance of GriT-DBSCAN, GriT-DBSCAN-LDF, and BLOCK-DBSCAN-FM, to BLOCK-DBSCAN and Wavefront \cite{gan2017hardness}.
For Wavefront, we use the binary code which is written in C++ \cite{binaryCode}.

\noindent\textbf{Influence of $\epsilon$.}
\begin{figure*}[!t]
\centering
\includegraphics[width=5in]{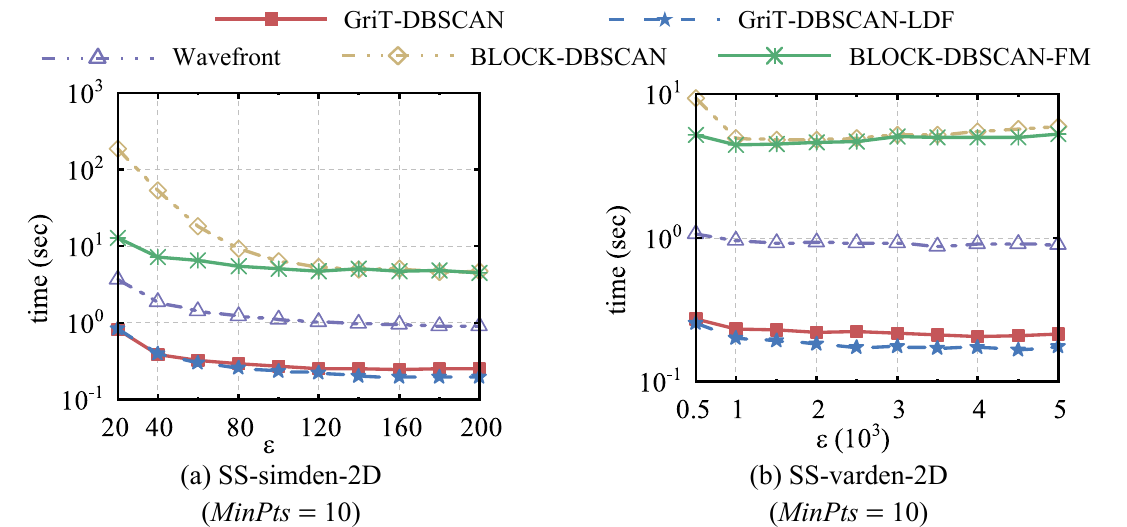}
\caption{Running time vs. $\epsilon$.}
\label{epsilon-2d}
\end{figure*}
We compare the performance of the five algorithms by varying $\epsilon$.
As shown in Fig. \ref{epsilon-2d}, the performance of the five algorithms improves as $\epsilon$ grows.
GriT-DBSCAN and GriT-DBSCAN-LDF outperform the other three algorithms in all cases, while Wavefront is comparable to BLOCK-DBSCAN and BLOCK-DBSCAN-FM.

\noindent\textbf{Influence of $MinPts$.}
\begin{figure*}[!t]
\centering
\includegraphics[width=5in]{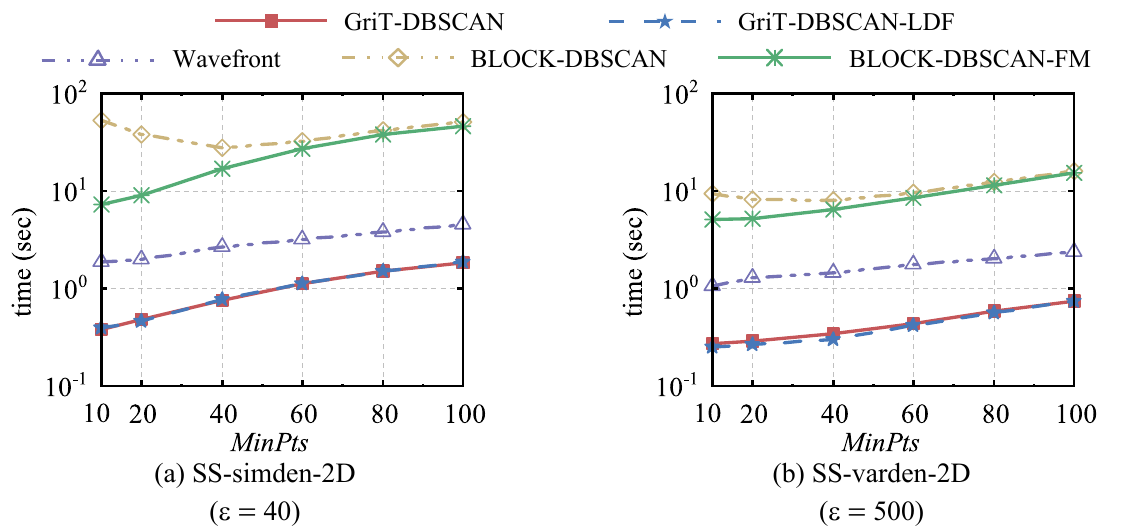}
\caption{Running time vs. $MinPts$.}
\label{minpts-2d}
\end{figure*}
We investigate the influence of $MinPts$ on the running time of each algorithm.
Fig. \ref{minpts-2d} shows the results.
The relative superiorities of all algorithms remain unchanged.
In addition, the performance of BLOCK-DBSCAN is consistent with the observations we made in the experiments for $d\geq 3$.

\noindent\textbf{Scalability with $n$.}
\begin{figure*}[!t]
\centering
\includegraphics[width=5in]{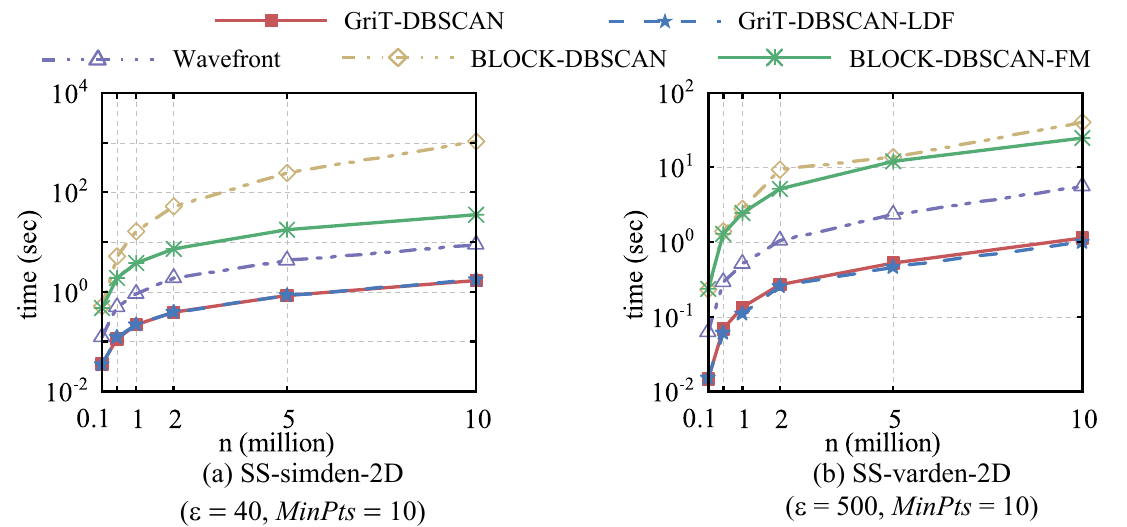}
\caption{Running time vs. $n$.}
\label{n-2d}
\end{figure*}
We vary the number of points in the data sets to examine how each algorithm scales with $n$.
The results are shown in Fig. \ref{n-2d}.
It can be seen that the relative superiorities of all algorithms remain the same.
However, in this set of experiments, BLOCK-DBSCAN is considerably slower than BLOCK-DBSCAN-FM on SS-simden-2D.
The reason is that a small $\epsilon$ leads to a large number of inner core blocks, which makes the cost of forming clusters expensive.
However, BLOCK-DBSCAN-FM uses a cover tree index, thus a small $\epsilon$ has little effect on its performance.

\subsection{Efficiency of Grid Tree}
\begin{figure*}[!t]
\centering
\includegraphics[width=6.5in]{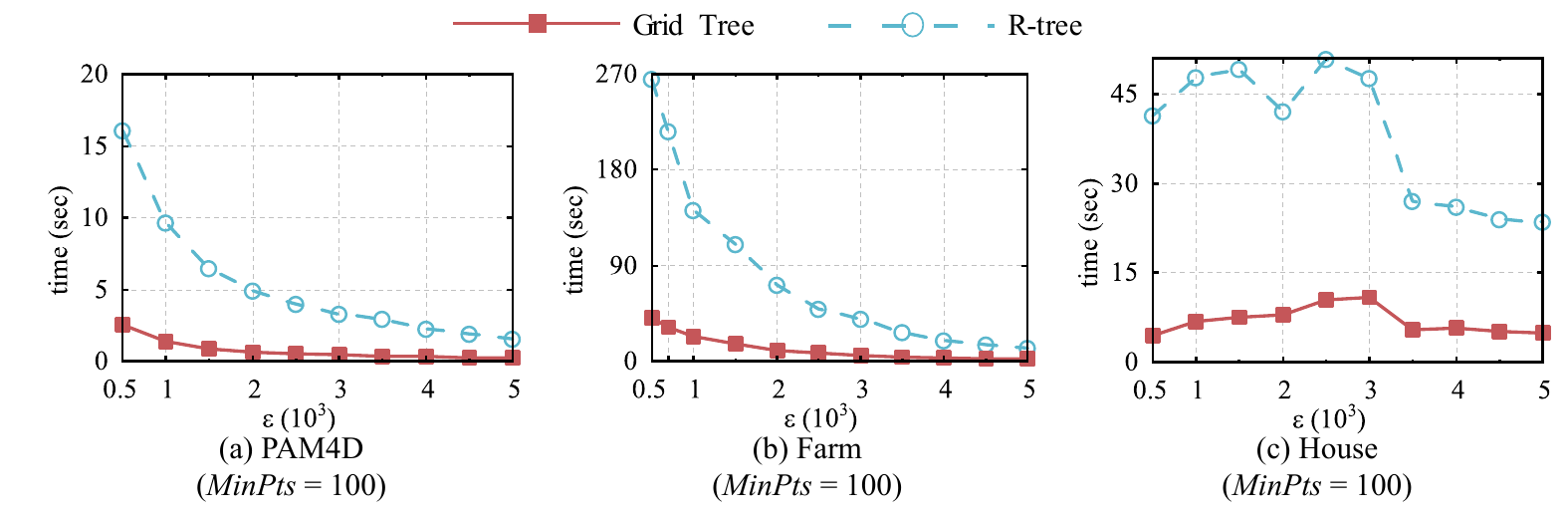}
\caption{The efficiency of grid tree.}
\label{gt}
\end{figure*}
To demonstrate the efficiency of grid tree, experiments are conducted on the three real-world data sets.
In this set of experiments, we compare the running time of grid tree and R-tree as neighboring grid query techniques by varying $\epsilon$ from 500 to 5000 and fixing $MinPts$ to the default value.
The results are shown in Fig. \ref{gt}.
In general, following an increase in $\epsilon$, both grid tree and R-tree become faster.
The reason is that the number of non-empty grids decreases with the increase of $\epsilon$, resulting in fewer neighboring grid queries.
On the other hand, the average number of non-empty neighboring grids increases with larger $\epsilon$, making the neighboring grid query more expensive, thereby increasing the running time.
Specifically, when $\epsilon$ is less than 3000, both grid tree and R-tree become slower with the increase of $\epsilon$ on the House data set.
In addition, it is obvious from Fig. \ref{gt} that grid tree significantly outperforms R-tree on all data sets.
We can conclude that grid tree clearly speeds up the neighboring grid query.

\section{Conclusions} \label{conclusions}
In this paper, we introduce a new exact DBSCAN algorithm with complexity linear to the number of data points called GriT-DBSCAN.
The key idea of GriT-DBSCAN is to utilize the spatial relationships among points to efficiently determine whether two core grids can be merged in the merging step.
More specifically, when judging whether two core grids can be merged, the trivial points in each grid are iteratively removed through the triangle inequality and the angle information, so as to reduce unnecessary distance calculations.
In addition, we introduce a grid tree to organize non-empty grids and an algorithm using it for efficient non-empty neighboring grids queries.
We prove theoretically that GriT-DBSCAN presents excellent improvement in terms of computational efficiency.

We also obtain two variants of GriT-DBSCAN, namely, GriT-DBSCAN-LDF and BLOCK-DBSCAN-FM.
To further improve the performance of GriT-DBSCAN, GriT-DBSCAN-LDF incorporates the union-find data structure and the low density first traverse strategy to reduce redundant merging checks: once two core grids are in the same set, it is unnecessary to check whether they can be merged.
By combining Algorithm \ref{fma} with BLOCK-DBSCAN, we obtain the second variant BLOCK-DBSCAN-FM.
Due to the accuracy and efficiency of Algorithm \ref{fma}, BLOCK-DBSCAN-FM obtains accurate clustering results with almost the same running time as BLOCK-DBSCAN or even less.

We conduct extensive experiments to evaluate the performance of GriT-DBSCAN and its two variants.
The results demonstrate that our algorithms are more efficient than existing algorithms.


%

%

%

\ifCLASSOPTIONcaptionsoff
  \newpage
\fi



\bibliographystyle{IEEEtran}
%

%

%





\end{document}